\documentclass[aps,prl,reprint,superscriptaddress]{revtex4-2}
\usepackage[colorlinks,citecolor=blue,bookmarks=true,breaklinks=true,linktocpage=true]{hyperref}
\usepackage{amsmath,amsfonts,amsthm}
\usepackage{xcolor}
\usepackage{braket}
\usepackage[capitalise]{cleveref}
\usepackage{graphicx}
\usepackage{ragged2e}

\newtheorem{lemma}{Lemma}
\newtheorem{thm}{Theorem}
\newtheorem{cor}{Corollary}
\newtheorem*{remark}{Remark}

\theoremstyle{definition}
\newtheorem{defi}{Definition}

\crefname{section}{Sec.}{Sec.}

\let\originalleft\left
\let\originalright\right
\renewcommand{\left}{\mathopen{}\mathclose\bgroup\originalleft}
\renewcommand{\right}{\aftergroup\egroup\originalright}

\begin{document}

\setlength{\belowdisplayskip}{5pt} \setlength{\belowdisplayshortskip}{2pt}
\setlength{\abovedisplayskip}{5pt} \setlength{\abovedisplayshortskip}{2pt}

\title{Time Independence Does Not Limit Information Flow. II. The Case with Ancillas}

\author{T.~C.~Mooney}
\email{tmooney@umd.edu}

\affiliation{Joint Center for Quantum Information and Computer Science, NIST/University of Maryland, College Park, Maryland 20742, USA}
\affiliation{Joint Quantum Institute, NIST/University of Maryland, College Park, Maryland 20742, USA}

\author{Dong~Yuan}
\affiliation{Center for Quantum Information, IIIS, Tsinghua University, Beijing 100084, People's Republic of China}
\affiliation{JILA, University of Colorado Boulder, Boulder, Colorado 80309, USA}

\author{Adam Ehrenberg}
\affiliation{Joint Center for Quantum Information and Computer Science, NIST/University of Maryland, College Park, Maryland 20742, USA}
\affiliation{Joint Quantum Institute, NIST/University of Maryland, College Park, Maryland 20742, USA}

\author{Christopher~L.~Baldwin}
\affiliation{Department of Physics and Astronomy, Michigan State University, East Lansing, Michigan 48824, USA}

\author{Alexey~V.~Gorshkov}
\affiliation{Joint Center for Quantum Information and Computer Science, NIST/University of Maryland, College Park, Maryland 20742, USA}
\affiliation{Joint Quantum Institute, NIST/University of Maryland, College Park, Maryland 20742, USA}

\author{Andrew~M.~Childs}
\affiliation{Joint Center for Quantum Information and Computer Science, NIST/University of Maryland, College Park, Maryland 20742, USA}
\affiliation{Department of Computer Science, University of Maryland, College Park, Maryland 20742, USA}
\affiliation{Institute for Advanced Computer Studies, University of Maryland, College Park, Maryland 20742, USA}

\begin{abstract}
    While the impact of locality restrictions on quantum dynamics and algorithmic complexity has been well studied in the general case of time-dependent Hamiltonians, the capabilities of time-independent protocols are less well understood. Using clock constructions, we show that the light cone for time-independent Hamiltonians, as captured by Lieb-Robinson bounds, is the same as that for time-dependent systems when local ancillas are allowed. More specifically, we develop time-independent protocols for approximate quantum state transfer with the same run-times as their corresponding time-dependent protocols. Given any piecewise-continuous Hamiltonian, our construction gives a time-independent Hamiltonian that implements its dynamics in the same time, up to error $\varepsilon$, at the cost of introducing a number of local ancilla qubits for each data qubit that is logarithmic in the number of qubits, the norm of the Hamiltonian and its derivative (if it exists), the run time, and $1/\varepsilon$. We apply this construction to state transfer for systems with power-law-decaying interactions and one-dimensional nearest-neighbor systems with disordered interaction strengths. In both cases, this gives time-independent protocols with the same optimal light-cone-saturating run-times as their time-dependent counterparts.
\end{abstract}
\maketitle

\textit{Introduction.}---Locality constraints underlie many aspects of quantum many-body dynamics that have received significant attention in recent years, such as topological order \cite{Bravyi_2006,Bravyi_2010,Hastings_2005} and many-body localization \cite{Nandkishore_2015, Abanin_2019,Altman_2015,Imbrie_2016}, among others. Locality constraints also play a role in the practical implementation of quantum computers---while most idealized theoretical models assume all-to-all connectivity for gates, real quantum computers often have quite restricted interaction graphs, and taking this into account can introduce polynomial overhead to implement arbitrary nonlocal gates, which can potentially eliminate polynomial quantum advantage in real systems \cite{Hirata2009,rosenbaum2013,Beals_2013,brierley2016,Bapat_2023}. 

One of the key theoretical results for systems with locality constraints is the Lieb-Robinson bound, which limits how fast quantum information can propagate \cite{Lieb_robinson_72,_Anthony_Chen_2023}. Originally proven for translation-invariant spin systems with finite-range interactions, they have been extended to accommodate long-range interactions \cite{Hastings_2006,Tran_2020_hierarchy,Tran_2021_LightCone,Foss_Feig_2015,Else_2020}, disorder \cite{Baldwin_2023,baldwin2024subballistic, Gebert_2022,chen2021OTOC, Hamza_2012}, bosons \cite{Chao_2022_boson, Kuwahara_2021, Eisert_2009, Schuch_2011}, and various other settings.

These bounds have been instrumental in providing rigorous proofs of various results in quantum many-body physics, such as the existence of the thermodynamic limit \cite{nachtergaele2011thermodynamic}, the stability of topological phases \cite{Bravyi_2006,Bravyi_2010,Hastings_2005}, quasiparticle scattering \cite{Naaijkens_2017,bachmann2015,Haegeman_2013}, and the higher-dimensional Lieb-Schultz-Mattis theorem \cite{Hastings_2004}. Additionally, they give bounds on the speed at which various quantum protocols, such as state transfer and Greenberger–Horne–Zeilinger (GHZ) state preparation, can be performed in specific models of quantum computing. Saturating these bounds has been an active direction of research \cite{Tran_2020_hierarchy, Eldredge2017, guo2024experimental}.

However, nearly all of the existing Lieb-Robinson-type bounds have applied to Hamiltonians with arbitrary time dependence. Time-independence---or equivalently, energy conservation---is a significant restriction in the power of a computational model, so it is natural to ask how it affects the capabilities of quantum many-body systems. 

In this paper, we develop a locality-preserving method to remove time dependence for piecewise-continuous time-dependent protocols, at the cost of a number of local ancillas for each data qubit that scales logarithmically in the allowable error $\varepsilon$, system size $N$, protocol run-time $T$, and quantities $h$ and $h_1,$ defined below in \cref{eq:hh1defi}, that represent the strength and rate of change of the Hamiltonian, respectively. Our approach modifies the clock construction of Watkins, Wiebe, Roggero, and Lee (WWRL) \cite{watkins2024time} that, given a base time-dependent Hamiltonian, outputs a time-independent Hamiltonian driven by a global clock whose time-evolution has rigorous guarantees on the deviation from that of the base Hamiltonian. The clock controls all time-dependent terms and is driven forward at a constant speed by a discretized momentum operator. We replace the one global clock with many local clocks (one per site) to preserve the locality structure of our system, having each clock driven by its own momentum operator, and letting a single clock control any given interaction abutting it, incurring a modest increase in local clock dimension (in addition to the overall increase due to changing one global clock to $N$ local ones). We call this modified clock Hamiltonian the localized WWRL construction. While one might worry that the different clocks desynchronize during the course of the time evolution, we prove that this effect can be made negligible and that the performance guarantees from the original WWRL construction continue to hold.

Since time-dependent Hamiltonians that saturate a variety of Lieb-Robinson bounds are known, this localized WWRL construction yields corresponding time-independent Hamiltonians that also saturate the bounds.
In other words, \textit{if allowed logarithmically many ancillas per site, Lieb-Robinson bounds cannot be tightened by assuming time-independence.}

As examples, we ``staticize'' (i.e., make time-independent) state-transfer protocols for systems with long-range interactions and disordered qubit chains from Refs.~\cite{tran2021optimal,yin2024ghz} and~\cite{Baldwin_2023}, respectively. These are two scenarios for which Lieb-Robinson bounds with non-linear light cones have been derived (i.e., information propagates across distance $r$ in time $t \propto r^z$ for $z\neq 1$). In the long-range case, we have $z<1$, and in the disordered case, we have $z>1$. These light cones are saturated via the aforementioned time-dependent protocols, and thus via our staticized versions as well. While there has been some prior investigation of time-independent protocols \cite{Bapat_2022, Markiewicz2009Perfect,Hermes2020Dimensionality,Bose2007Quantum,Bose2003Quantum}, we are not aware of prior work that realizes optimal time-independent state-transfer protocols in the settings we staticize, namely the long-range-interacting and disordered cases. While these examples illustrate the capabilities of our construction, the localized WWRL procedure works for any piecewise-continuous protocol, including in particular quantum circuits.

\textit{Setup.---}Consider an $N$-qubit Hilbert space $\mathcal{H}=\left(\mathbb{C}^{2}\right)^{\otimes N},$ with the qubits located on the vertices $V$ of a graph with edges $E$ \footnote{We note we could also consider an $N$-qudit system for constant $d$; for simplicity, though, we will restrict our work to qubits.}. We are given some time-dependent Hamiltonian $H(t):=\sum_{e\in E} H_e(t){ + \sum_{v\in V}H_v(t)},$ defined for $t\in[0,T],$ {and such that $\sup_{t\in [0,T]}\Vert H_e(t)\Vert \leq w(e)$ for all $e\in E.$ The $w(e)$ encode the specific locality structure of our Hamiltonian; for instance, we can consider long-range-interacting models by letting $V\subset \mathbb{R}^{d}$ and $w(\{\mathbf{i},\mathbf{j}\}):=\Vert \mathbf{i}-\mathbf{j}\Vert^{-\alpha}$ for some $\alpha>0.$ Note that, for now, there is no restriction on the magnitudes of the single-site terms.} While we will later relax these assumptions, for now suppose that $H$ is differentiable, that $H(0) = H(T)$ and $\dot{H}(0) = \dot{H}(T) = 0$, and that the following quantities are finite:
\begin{align}\label{eq:hh1defi}
h:=\sum_{{ x}\in  E{ \cup V}}\max_t \Vert H_{ x}(t)\Vert&,&  h_1:&=\sum_{{ x}\in E{ \cup V}} \max_t \Vert \dot{H}_{ x}(t)\Vert.
\end{align}
We wish to construct some time-independent Hamiltonian $\overline{H}$, with $\mathcal{O}(\log(TN))$ local ancilla qubits per site and respecting the same locality constraints as $H(t)$, such that $e^{-iT\overline{H}}\approx U(T),$ where the latter is the time-evolution operator generated by $H(t)$ from $0$ to $T$.

We do this by modifying the WWRL construction from Ref.~\cite{watkins2024time}.
Specifically, we augment the Hilbert space at every vertex $v$ with an ancillary clock qudit of dimension $N_c\in\mathbb{N}$ that we specify later. The augmented Hilbert space is $\mathcal{H}_{\mathrm{aug}}:=\left(\mathbb{C}^2\otimes \mathbb{C}^{N_c}\right)^{\otimes N}$. Since $H(0)=H(T)$ and $\dot{H}(0)=\dot{H}(T)=0,$ we can define the Hamiltonian for all $t \in \mathbb{R}$ by making it periodic with period $T$---take $H(t)$ to refer to this extension from now on.

\textit{Details of the Localized WWRL Construction.}---We have $N_c$ clock states, each corresponding to a time step of duration $\delta:= T/N_c$ within a total time period of $T$. For later technical convenience, we let $N_c=N_pN_q$, with $N_p,N_q\in\mathbb{N}$, and $\tau:=T/N_p=N_q\delta$. We use this two-level coarse-graining of the discretization to control the two main sources of error in the construction described below.

The intuition for the modified WWRL construction is as follows. We construct a time-independent Hamiltonian $\overline{H}$ on $\mathcal{H}_{\mathrm{aug}}$ which is the sum of two non-commuting terms. The first term, $\boldsymbol{\Delta},$ is the sum of discretized momentum operators on each clock, which drives the clocks forward. The second term, $C(H),$ implements controlled applications of the Hamiltonian terms $H_e$, depending on the clock state of one of the vertices in each edge. Thus, the momentum operators keep the clocks moving, while the controlled Hamiltonian terms ensure the correct evolution is happening to all data qubits. See \cref{fig:cartoon} for a schematic of the construction.
 
\begin{figure}
    \centering
    \includegraphics[height=0.6\linewidth]{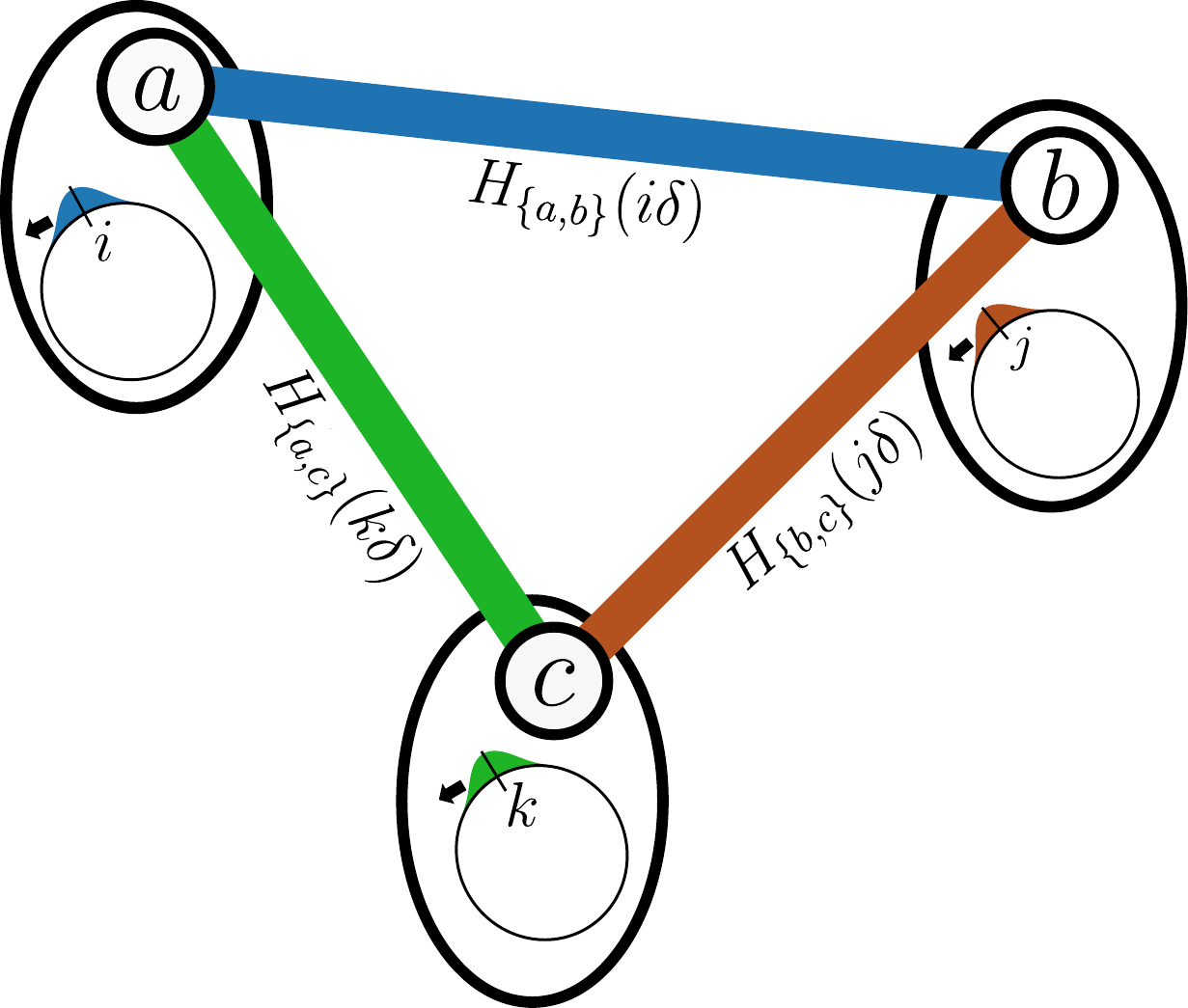}
    \caption{A cartoon of the localized WWRL construction acting on 3 qubits. The edges' colors corresponds to the color of their controlling clocks. The clock states are set to be Gaussian wave packets. Note that there is no a priori requirement that edges are controlled by different clocks, but in this cartoon, that happens to be the case.
    }
    \label{fig:cartoon}
\end{figure}
More concretely, for any $v\in V$, let $U_v:= \sum_{k\in \mathbb{Z}_{N_c}}(\ket{k+1}\bra{k})_v,$ where $\ket{k}$ denotes the state of the clock associated to qubit $v$. Then let $\boldsymbol{\Delta}:=\sum_{v\in V}\Delta_v$, where
\begin{equation}
        \Delta_v := i\frac{U_{v}-U^{\dagger}_{v}}{2\delta }.
\end{equation}
This is our discretized total momentum operator. 
Next, we implement controlled applications of the $H_e$ terms. To preserve the locality structure of the base Hamiltonian, we associate our clocks  with vertices, not edges. This means that for each interaction term $H_e$ we have to choose one of the vertices in $e$ on which to control the interaction. We encode this (arbitrary) choice into a function $\rho\colon E\to V$ with $\rho(e)\in e$ for all $e.$ With this defined, let \begin{align}C(H):=&\sum_{k\in\mathbb{Z}_{N_c}}\sum_{e\in E}H_e(k\delta )\otimes (\ket{k}\bra{k})_{\rho(e)}\nonumber\\\qquad\qquad& { +\sum_{k\in\mathbb{Z}_{N_c}} \sum_{v\in V}H_v(k\delta)\otimes (\ket{k}\bra{k})_v}.\end{align}
   The localized WWRL Hamiltonian is then \begin{equation}\overline{H}:= \boldsymbol{\Delta}+C(H).\end{equation}

Next we introduce the initial states of the clock registers. To balance the localization in position and momentum space, we choose a Gaussian wave packet. Letting $\vert \cdot\vert_c:=\min\{\vert \cdot\vert, \vert N_c-\cdot\vert\}$ be the distance on $\mathbb{Z}_{N_c},$ take $\ket{\phi_j}$ to be the one-dimensional Gaussian state centered at $j\in\mathbb{Z}_{N_c}$ with standard deviation $\sigma/\delta$ (or, alternatively, a continuum Gaussian centered at $j\delta$ with standard deviation $\sigma$ and discretization $\delta$): \begin{equation}\ket{\phi_j}:=\frac{1}{\sqrt\mathcal{N}}\sum_{l\in\mathbb{Z}_{N_c}}e^{-\frac{\delta^2\vert l-j\vert_c^2}{\sigma^2}}\ket{l},\end{equation} with $\mathcal{N}$ chosen so $\Vert \ket{\phi_j}\Vert=1$. For $\mathbf{j}\in\mathbb{Z}_{N_c}^{N},$ then take $\ket{\Phi_{\mathbf{j}}}:=\bigotimes_{\mathfrak{j}\in\mathbf{j}}\ket{\phi_{\mathfrak{j}}}.$ We initialize the clocks to $\ket{\Phi_{\mathbf{0}}}.$ For this construction to work, $\sigma$ must be chosen carefully, as discussed in \cref{cor:ancillacount} of the supplemental material. 

We prove that $e^{-i\overline{H}T}\ket{\psi}\ket{\Phi_{\mathbf{0}}}\approx (U(T)\ket\psi)\otimes \ket{\Phi_{\mathbf{0}}}.$
{This follows from the first-order Trotter product formula (with time step $\delta$) and two properties:}
\begin{enumerate}
    \item 
    $e^{-im\delta \boldsymbol{\Delta}}\ket{\Phi_{\mathbf{r}}}\approx \ket{\Phi_{\mathbf{r}+m\mathbf{1}}}$ for suitably small $\delta>0,$ $m\in \mathbb{Z}.$ This holds because $\boldsymbol{\Delta}$ is approximately the total momentum operator, and thus approximately generates translation of all clocks.
    \item 
    $e^{-iC(H)\eta}\ket\psi\ket{\Phi_{j\mathbf{1}}}\approx (e^{-iH(j\delta)\eta}\ket\psi)\otimes \ket{\Phi_{j\mathbf{1}}}$ for sufficiently small $\vert \eta\vert,$ $\sigma/\tau.$ This holds because $\ket{\boldsymbol{\Phi}_{j\boldsymbol{1}}}$ has each clock localized around $j\delta,$ and $C(H)$ acts as $H(j\delta)$ for clocks in that state. 
\end{enumerate}
These two properties are rigorously stated in the supplemental material in Lemmas \ref{lemma:translation_error} and \ref{lemma:controlledapperro}, respectively.

We then combine the above properties to finish the justification. We first note that $e^{-i\overline{H}T}=\left(e^{-i\tau \overline{H}}\right)^{N_p}.$ By the Trotter product {error bounds and \cref{lemma:commutator}}, 
when $\tau^2 h_1$
 is small {(and, thus, by \cref{lemma:commutator}, $\tau^2 \Vert [\boldsymbol{\Delta},C(H)]\Vert$ is small as well),} $e^{-i\tau\overline{H}}\approx e^{-i\tau\mathbf{\Delta}}e^{-i\tau C(H)}$. Therefore,
\begin{align}
    e^{-i\tau\boldsymbol{\Delta}} e^{-i\tau C(H)} \ket\psi\ket{\Phi_{lN_q\mathbf{1}}} &\approx e^{-i\tau\boldsymbol{\Delta}} (e^{-i\tau H(l\tau)}\!\ket\psi)
    \!\ket{\Phi_{lN_q\mathbf{1}}}\nonumber\\
    &\approx (e^{-i\tau H(l\tau)}\!\ket\psi)\!\ket{\Phi_{(l+1)N_q\mathbf{1}}}.
\end{align}
Applying this repeatedly, we have
\begin{equation}
    e^{\!\!-i\overline{H}T}\!\!\ket\psi\!\!\ket{\Phi_{\mathbf{0}}}\! \approx\! \left(e^{\!-i\tau H((N_p-1)\tau)}\!\cdots\! e^{\!-i\tau H(0)}\ket\psi \right)\!\!\ket{\Phi_{\mathbf{0}}},
\end{equation}
where we wrap back around to $\ket{\Phi_{\mathbf{0}}}$ due to the periodicity.

Now, $e^{\!-i\tau H((N_p-1)\tau)}\cdots e^{\!-i\tau H(0)}$ is in turn a first-order approximation of $U(T),$ with low error if $\frac{T^2}{N_p}\max_t\Vert \dot{H}(t)\Vert\leq \frac{T^2h_1}{N_p}$ is small, so 
\begin{equation}
    e^{-i\overline{H}T}\ket\psi\ket{\Phi_{\mathbf{0}}} \approx \left(U(T)\ket\psi \right)\!\ket{\Phi_{\mathbf{0}}}.
\end{equation}
It may be surprising that the clocks evolve without any back-action from the data qubits, since the two are coupled through $C(H)$. Yet when the Hamiltonian is time-independent, $C(H)$ commutes with $\Delta$ and hence there is no back-action. More generally, the smallness of $h_1T^2/N_p$ implies that $C(H)$ and $\Delta$ approximately commute, so the back-action is negligible over the run-time, {as shown in \cref{lemma:commutator} of the} supplemental material. {Carrying out this analysis quantitatively, we obtain explicit requirements on the relevant parameters, and in particular, show that to have an overall error of at most $\varepsilon,$ it suffices to take
\begin{align}\label{eq:clockscaling}
    N_c&=\Theta\left(\frac{\sqrt{N}h_1^2T^4}{\varepsilon^3}\log\frac{hT\sqrt{N}}{\varepsilon}\right).
\end{align} For the details of this calculation, see \cref{cor:ancillacount} and Eqs.~(\ref{eq:asymptoticancilla1}--\ref{eq:asymptoticancilla2}) in the supplemental material.}

\textit{Mollifier Convolution.}---The preceding discussion was for differentiable periodic Hamiltonians.
To staticize a general piecewise-continuous time-dependent Hamiltonian $H_{\text{base}}(t)$, we approximate it with a smooth one by convolving it with a specific compactly supported smooth mollifier.
We begin by defining a Hamiltonian $\hat{H}_{\text{base}}(t):= 2H_{\text{base}}(2(t-T/4))\mathbf{1}[t\in[T/4,3T/4]]$ which generates the same evolution as $H_{\text{base}}$ over time $T$ but is supported on $[T/4,3T/4]$ instead. We then convolve $\hat{H}_{\text{base}}(t)$ with a mollifier $\phi_s(t)$ compactly supported on $[-s,s]$ for some $s\leq T/4,$ in a very similar manner to Ref.~\cite{Poulin_2011}, to get the smoothed Hamiltonian
\begin{align}
    \tilde{H}^{(s)}(t):= (\phi_s*\hat{H}_{\text{base}})(t).
\end{align} Since $\tilde{H}^{(s)}$ is differentiable and (owing to $s\leq T/4$) $\dot{\tilde{H}}^{(s)}(0)=\dot{\tilde{H}}^{(s)}(T)=\tilde{H}^{(s)}(0)=\tilde{H}^{(s)}(T)=0,$ we can apply our localized WWRL construction to it. It remains only to choose $s$ such that the evolution under $\tilde{H}^{(s)}$ is close to that under $H_{\text{base}}.$

The difference in operator norm of the evolutions generated by $\tilde{H}^{(s)}$ and $H_{\text{base}}$ is at most $3Tsh^2/2$ (see \cref{lemma:convolutionerror}). Further, it is straightforward to show that $\tilde{h}:=\sum_{ x}\max_t \Vert \tilde{H}_{ x}(t)\Vert \leq 2h$ and $\tilde{h}_1:=\sum_{ x}\max_t \Vert \dot{\tilde{H}}_{ x}(t)\Vert \leq 3h/2\nu s$ (\cref{lemma:h1tilde}) for $\nu\approx0.222${. We can plug these values} into \cref{thm:errorbounds} to {obtain the error bound of \cref{cor:mollconvclockerror}. We then use this error bound to determine {first} how {small $s$ must be and then how} large the clock dimension should be to get a total error of $\varepsilon.$}

To determine how small $s$ needs to be, we demand that the approximation error from the convolution is at most $\varepsilon/2$, so $s \leq \frac{\varepsilon}{3h^2T}$. Letting $s = \frac{\varepsilon}{3h^2T}$, we have $\tilde{h}_1 \leq 9h^3T/2\nu\varepsilon$. Thus, taking $H=\tilde{H}^{(\varepsilon/3h^2T)}$ for our localized WWRL procedure and then plugging our bounds on $\tilde{h},\tilde{h}_1$ into {\cref{eq:clockscaling}}, we obtain an $\varepsilon$-error staticization of a general piecewise-continuous protocol with local clock dimension 
\begin{align}\label{eq:clockdim}
        N_c &= \mathcal{O}\left(\frac{\sqrt{N}h^6T^6}{\varepsilon^5}\log\frac{hT\sqrt{N}}{\varepsilon}\right).
\end{align}

\textit{Application: Long-Range-Interacting Systems.}---It is possible to generate entanglement super-ballistically for Hamiltonians with sufficiently slowly decaying power-law interactions. An optimal time-dependent protocol saturating the Lieb-Robinson bounds was found in Ref.~\cite{tran2021optimal}. Consider qubits arranged in a hypercubic lattice in $d$ dimensions. Let $\alpha\in (2d,2d+1);$ {the analysis is similar, though more technically involved, for the $\alpha\in(d,2d)$ and $\alpha=2d$ regimes, and we get the same resulting ancilla dimension, \cref{eq:ancilladimpower}, up to substituting the appropriate time-dependent run-time. For details of these other regimes, as well as for $\alpha\in[0,d]$ (which staticizes a different underlying protocol \cite{yin2024ghz}), see \cref{section:LRProt} and \cref{section:SLRProt} of the supplemental material 
(note that, for $\alpha\geq 2d+1$, the light cone is linear so there already exist saturating time-independent protocols). We suppose} the lattice has side length ${\left({\left\lceil 3^{\frac{1}{\alpha-2d}}\right\rceil+1}\right)^q}$ for some $q\in\mathbb{N}$, meaning that $N=\left(\left\lceil 3^{\frac{1}{\alpha-2d}}\right\rceil+1\right)^{qd}$. Applying the localized WWRL construction, each of these data qubits is given a clock ancilla qudit located on the same site as the data qubit. We further impose the restriction that all two-site interactions $H_{\{\mathbf{i},\mathbf{j}\}}$ between sites located at $\mathbf{i}$ and $\mathbf{j}$ must satisfy $\Vert H_{\{\mathbf{i},\mathbf{j}\}}\Vert \leq \Vert \mathbf{i}-\mathbf{j}\Vert^{-\alpha}.$ We restrict single-site interactions to unit norm \footnote{We note that while this does diverge somewhat from the setup in Ref.~\cite{tran2021optimal}, which allowed single-site terms to have arbitrary norm, the contribution to the run time from restricting single-site norms is always negligible compared to the other steps, and thus the asymptotic scaling does not change.}. Our task is to transfer the unknown state $\ket{\psi}=a\ket{0}+b\ket{1}$ from the $\mathbf{0}$ site to the $(N^{1/d}-1)\mathbf{1}$ site, i.e., the furthest site from $\mathbf{0}.$ We assume all other sites have their data qubits in the $\ket0$ state. We do the transfer by encoding the initial state into a modified GHZ state $a\ket{\overline{0}}^{\otimes N}+b\ket{\overline{1}}^{\otimes N}$ and then running that procedure in reverse, but decoding the state onto the $(N^{1/d}-1)\mathbf{1}$ site as opposed to $\mathbf{0}.$ 

To do the encoding into the modified GHZ state, we staticize the procedure from Ref.~\cite{tran2021optimal}. To analyze the ancilla count, we first 
explain the original protocol. For brevity, as specified above, we only consider the regime where $\alpha\in (2d,2d+1)$. We divide the sublattice into a collection of progressively coarser hypercubes of qubits, each of size $V_j:= r_j^d$ where $r_j:= \left(\left\lceil 3^{\frac{1}{\alpha-2d}}\right\rceil+1\right)^{j}.$ The protocol recursively builds GHZ states on the larger hypercubes out of the smaller hypercubes. We focus our discussion in this section on the recursive step.

Let us assume we can construct the GHZ states for the smaller hypercubes.
The protocol then proceeds as follows: first, we construct the modified GHZ states on the smaller hypercubes. Then, we choose one of these hypercubes to be our control, and perform controlled phase gates (in the logical $\{\ket{\bar{0}}:=\ket{0}^{\otimes V_j},\ket{\bar{1}}:=\ket{1}^{\otimes V_j}\}$ subspaces) on all the other hypercubes, obtaining the state $a\ket{\bar{0}}\ket{\bar{+}}^{\otimes(N/V_j-1)}+b\ket{\bar{1}}\ket{\bar{-}}^{\otimes(N/V_j-1)}$. To perform a logical Hadamard gate on those hypercubes, we then invert the smaller-hypercube GHZ-construction protocol, returning the state to 
\begin{equation}
  a\ket{\overline{0}}\left(\ket+\ket{0}^{\otimes V_j-1}\right)^{\otimes (\frac{N}{V_j}-1)}+b\ket{\overline{0}}\left(\ket-\ket{0}^{\otimes V_j-1}\right)^{\otimes (\frac{N}{V_j}-1)}
\end{equation} 
for specific sites inside the hypercubes, and all other sites to $\ket{0}.$ We then perform Hadamard gates on those specific sites, and finally apply the hypercube GHZ protocol once again, completing the logical Hadamards.

To analyze the performance of this model's staticization, we calculate $h.$ 
As the Hamiltonian of each step is just a rescaled sum of commuting projectors, the calculation is significantly simplified. For details, see {\cref{section:LRProt} of the} supplementa{l} material, which shows that $h$ is $\Theta(N).$

Plugging the scaling of $h$ into \cref{eq:clockdim}, we find 
\begin{equation}
N_c=\mathcal{O}\left( \frac{N^{13/2}T^6}{\varepsilon^{5}}\log\frac{ N^{3/2}T}{\varepsilon}\right).\label{eq:ancilladimpower}\end{equation} We then take the logarithm to find the number of ancilla qubits per data qubit. We reiterate that, in \cref{sec:bumpfunc} of the supplemental material, we use a different smoothing technique to get a better scaling of $N_c$ with $N.$

\textit{Application: Disordered Spin Chain.}---On the other hand, for nearest-neighbor systems with disordered interaction strengths, when the low-strength tail of the interaction distribution is sufficiently heavy, the fastest one can generate entanglement is sub-ballistically. The SWAP-based time-dependent state-transfer protocol was shown to saturate the Lieb-Robinson bounds in Ref.~\cite{Baldwin_2023}. Consider a one-dimensional chain of qubits with nearest-neighbor couplings $\Vert H_{\{i,i+1\}}\Vert\leq J_{\{i,i+1\}}.$ Let $\mu_N(J):=\frac{1}{N}\sum_{i=0}^{N-1}\mathbf{1}[J_{\{i,i+1\}}\leq J].$ We assume that $\mu_N$ has a well defined limit $\mu$ \footnote{we note the technical requirement that this limit must not only satisfy $\lim_{N\to\infty}\mu_N(J)=\mu(J)$ but also $\lim_{N\to\infty}\frac{\mu_N(JN^{-\beta})}{\mu(JN^{-\beta})}=1$ for all $J>0,\beta\in[0,\alpha^{-1})$ \cite{Baldwin_2023}} that is $\Theta(J^\alpha)$ for some $\alpha \geq 0$ as $J\to0^+,$ and that, as $N\to\infty$,  the maximum $J$ is bounded and the minimum $J$ decays slower than $N^{-\beta}$ for all $\beta>1/\alpha.$

Then, we perform state transfer between site $0$ and site $N-1$ by performing SWAP gates along all edges between them. A SWAP gate across edge $\{i,i+1\}$ can be implemented by running a norm-$J_{\{i,i+1\}}$ Hamiltonian for time $\Theta(J_{\{i,i+1\}}^{-1}).$ Prior work \cite{Baldwin_2023} has shown that the total run time $T_N$ for this protocol satisfies
\begin{equation}
\lim_{N\to\infty}\frac{T_N}{N^z}= \begin{cases}0& z>\max\{1,\alpha^{-1}\},\\\infty&z<\max\{1,\alpha^{-1}\}.\end{cases}
 \end{equation}
Finally, note that $h=\Theta\left(\max_{i=0}^ {N-1}J_{\{i,i+1\}}\right)=\Theta(1)$.

Thus, letting $z_{\mathrm{c}}:=\max\{1,\alpha^{-1}\}$, and plugging into \cref{eq:clockdim}, we have 
\begin{equation}
N_c=\mathcal{O}\left(\frac{N^{0.5+6.6z_{\mathrm{c}}}}{\varepsilon^5}\log\frac{N^{0.5+1.1z_{\mathrm{c}}}}{\varepsilon}\right).
\end{equation} 
Once again, in \cref{sec:bumpfunc} we get better scaling with $N$. 

\textit{Discussion.}---In this work, we created a staticization procedure that is local and uses logarithmically many local ancillas. We showed that, with the aid of such ancillas, time-independent Hamiltonians can saturate the same Lieb-Robinson light cone as time-dependent protocols.
However, it remains open whether we can reduce the local ancilla count to constant. Equivalently, this would mean embedding both data and ancilla qubits into a lattice in $\mathbb{R}^d$ such that all qubits are distance $\Omega(1)$ from each other in the $N\to\infty$ limit, while still satisfying the interaction-strength constraints. In our companion paper \cite{yuan_2025}, we do so for long-range free-particle Hamiltonians. Unfortunately, for general Hamiltonians, the localized WWRL protocol seems hard to extend to low-ancilla regimes, since to preserve its accuracy as the number of qubits increases, the protocol requires the number of clock steps per site to increase. 

However, we can still consider improvements to this construction. It may be possible to implement the clock for the specific protocols staticized above in a locality-respecting way while reducing the redundancy introduced by having $N$ local clocks. Additionally, perturbative gadgets \cite{Cubitt_2018} might be able to decrease ($k$-)locality and thus eliminate the {higher}-locality interactions in the construction, potentially resulting in a completely 2-local Hamiltonian that may be more practical to implement experimentally. However, this might induce a significant slowdown due to the necessity of encoding the 3-local interactions into the low-energy subspace of the Hamiltonian.

It is worth emphasizing the generality of our approach. While we focused on two specific settings with saturable non-linear light cones, the localized WWRL construction can be applied to 
\textit{any} piecewise-continuous Hamiltonian (although it is only preferable to the standard WWRL construction when wanting to preserve the underlying locality of the base quantum system). Namely, this includes any circuit-based quantum algorithm, many quantum annealing protocols, and entangled-state preparation methods for, e.g., quantum sensing applications. Investigating the specifics of such applications is an exciting direction for future research.

\begin{acknowledgments}
\textit{Acknowledgments.}---
We thank Dhruv Devulapalli for helpful discussions. T.C.M., A.E., A.V.G., and A.M.C.~were supported in part by the DoE ASCR Quantum Testbed Pathfinder program (awards No.~DE-SC0019040 and No.~DE-SC0024220) and NSF QLCI (award No.~OMA-2120757). T.C.M., A.E., A.V.G., and A.M.C.~also acknowledge support from the U.S.~Department of Energy, Office of Science, Accelerated Research in Quantum Computing, Fundamental Algorithmic Research toward Quantum Utility program (FAR-Qu). T.C.M., A.E., and A.V.G.~were also supported in part by the NSF STAQ program, AFOSR MURI, ARL (W911NF-24-2-0107), DARPA SAVaNT ADVENT, and NQVL:QSTD:Pilot:FTL.  T.C.M., A.E., and A.V.G.~also acknowledge support from the U.S.~Department of Energy, Office of Science, National Quantum Information Science Research Centers, Quantum Systems Accelerator (QSA). C.L.B.~was supported by start-up funds from Michigan State University. D.Y. acknowledges support from the National Natural Science Foundation of China (Grant No.~123B2072).
\end{acknowledgments}

\bibliography{references}
\onecolumngrid
\newpage

\setcounter{secnumdepth}{1}
\setcounter{section}{0}
\renewcommand{\thesection}{S\arabic{section}}
\setcounter{thm}{0}
\renewcommand{\thethm}{S\arabic{thm}}
\setcounter{cor}{0}
\renewcommand{\thecor}{S\arabic{cor}}
\setcounter{lemma}{0}
\renewcommand{\thelemma}{S\arabic{lemma}}
\setcounter{equation}{0}
\renewcommand{\theequation}{S\arabic{equation}}
\setcounter{table}{0}
\renewcommand{\thetable}{S\arabic{table}}
\setcounter{figure}{0}
\renewcommand{\thefigure}{S\arabic{figure}}
\setcounter{defi}{0}
\renewcommand{\thedefi}{S\arabic{defi}}

\center{\Large\textbf{Supplemental Material for ``Time independence does not limit information flow. II. The case with ancillas''}}
\vspace{1em}

 \raggedright
\setlength{\parindent}{20pt}
 \justifying

In this supplemental material, we provide rigorous theorems that support our discussion of the performance of the localized WWRL construction, and also give a more comprehensive discussion of the GHZ state preparation protocol for long-range interacting systems. To do so, in \cref{section:prelims}
, we first reiterate definitions of constructions from the main text as well as define some quantities that will clean up the presentation of the results in this supplemental material. In \cref{section:techlemma}, for completeness, we state and prove important technical lemmas that we then use in \cref{section:proofs} to rigorously prove the relation between ancilla dimension and error we discuss heuristically in the main text. {In \cref{sec:mollifconv}, we introduce a way to approximate general piecewise-time-continuous Hamiltonians with smooth ones via convolving with a compactly-supported smooth function, thereby extending the applicability of staticization to piecewise-continuous Hamiltonians.} In \cref{section:LRProt}, we introduce all of the technical details of the GHZ-state preparation protocol we staticize in the long-range setting for power-law exponent $\alpha\in (d,2d+1).$ In \cref{section:SLRProt}, we do likewise for the protocol when $\alpha \in [0,d].$ Finally, in \cref{sec:bumpfunc}, we introduce an alternate staticization procedure for piecewise time-independent Hamiltonians (a notable case being quantum circuits) that can in some cases, including our two specific examples, give better local ancilla dimension scaling, which we additionally derive.

\section{Preliminaries and Notation}\label{section:prelims}
In this section, we review all of the relevant constructions from the main text and reiterate some key definitions. 

First, recall that we are given a differentiable $T$-periodic Hamiltonian $H(t)$ as input. We let $N_c=N_pN_q$ for some natural numbers $N_p,N_q,$ and let $\delta := T/N_c$ and $\tau := T/N_p.$ We will also assume that $N_c$ is even.
\begin{defi}[Localized WWRL Hamiltonian]\label{defi:suppWWRL} Let $U_v:=\sum_{k\in \mathbb{Z}_{N_c}}(\ket{k+1}\bra{k})_v,$ \begin{equation}\Delta_v:= i\frac{U_v-U_v^\dagger}{2\delta},\end{equation} and \begin{equation}\boldsymbol{\Delta}:=\sum_{v\in V}\Delta_v.\end{equation} Further, let \begin{align}C(H):=& \sum_{k\in \mathbb{Z}_{N_c}}\sum_{e\in E}H_e(k\delta)\otimes (\ket{k}\bra{k})_{\rho(e)}\nonumber\\&\qquad{ +\sum_{k\in\mathbb{Z}_{N_c}}H_v(k\delta)\otimes (\ket{k}\bra{k})_{v}},\end{align} where we recall that for all $e\in E,$ $\rho(e)$ is some vertex contained in $e.$ Then, let \begin{equation}
\overline{H}:= \boldsymbol{\Delta}+C(H).\end{equation}
\end{defi}
\begin{defi}[Gaussian States]\label{defi:gauss}
     For $x\in \mathbb{Z}_{N_c},$ let $\vert x\vert_c:=\min\{\vert x\vert, \vert N_c-x\vert\}.$ Then, for $j\in\mathbb{Z}_{N_c},$ we let $\ket{\phi_j}\in\mathbb{C}^{N_c}$ be defined by 
     \begin{equation}
         \ket{\phi_j}:= \frac{1}{\sqrt{\mathcal{N}}}\sum_{l\in \mathbb{Z}_{N_c}}e^{-\frac{\delta^2}{\sigma^2}\vert l-j\vert_c^2}\ket{l},
     \end{equation} where $\mathcal{N}$ is defined so $\Vert\ket{\phi_j}\Vert = 1.$ Further, for $\mathbf{j}\in \mathbb{Z}_{N_c}^N,$ we let 
     \begin{equation}
         \ket{\Phi_{\mathbf{j}}}:= \bigotimes_{\mathfrak{j}\in\mathbf{j}}\ket{\phi_{\mathfrak{j}}}.
     \end{equation}
\end{defi}

\begin{defi}[Constants]\label{defi:constants}

Let \begin{align}
    A:&=\sqrt{1+3\frac{\delta}{\sigma}},\\
B:&=\sqrt{1-\left[\frac{\sigma}{T}e^{-\frac{T^2}{2\sigma^2}}+\frac{\delta}{\sigma}\right]},\\
 D:&=1 + 2^{-3/2}\frac{\delta}{\sigma} +\left(\frac{\delta}{\sigma}\right)^2,\\
 h:&= \sum_{ x\in E\cup V}\max_{t\in[0,T]}\Vert H_{ x}(t)\Vert, 
\end{align} and 
\begin{equation}
    h_1:= \sum_{{  x\in E\cup V}}\max_{t\in[0,T]}\Vert \dot{H}_{  x}(t)\Vert.
\end{equation}
\end{defi}

\begin{defi}[Partial Difference]\label{defi:partialdiff}
    Let $f:\mathbb{R}\supset S\to \mathbb{C}$ such that $[x-h,x+h]\in S.$ Then,
    \begin{equation}
        [D_hf](x):= \frac{f(x+h)-f(x-h)}{2h}.
    \end{equation}
\end{defi}
We define $D_h^kf$ by repeated composition of the partial difference operator.

\section{Technical lemmas}\label{section:techlemma}
For completeness, in this section, we collect technical results that we use to prove the main results of the paper. 

\subsection{Tensor Product Lemma}
In this subsection, we prove results bounding the difference of tensor powers of vectors in terms of the difference of the vectors, which we will use later to extend the WWRL bounds to the localized case.
\begin{lemma}\label{lemma:tensor}
    Let $\ket{\psi},\ket{\phi}\in \mathbb{C}^D$ satisfy $\langle \psi\vert\phi\rangle \in\mathbb{R}.$ Then, 
\begin{equation}
\left\Vert \ket{\psi}^{\otimes N}-\ket{\phi}^{\otimes N}\right\Vert \leq \sqrt{N}\Vert\!\ket{\psi}-\ket\phi\!\Vert .\end{equation}
\end{lemma}
\begin{proof}
    Let \begin{equation}\ket{\psi} = a \ket{\phi}+b\ket\perp,\end{equation} 
    where $a,b\in[0,1],$ $a^2+b^2=1,$ and $\langle{\perp}\vert \phi\rangle = 0.$ We have $\Vert \!\ket\psi-\ket\phi\!\Vert = \sqrt{2(1-a)}.$ Further,
    \begin{equation}\ket\psi^{\otimes N}= a^N \ket\phi^{\otimes N}+c \ket{\perp'},\end{equation} where $a^{2N}+c^2=1,$ and $\bra{\perp'}(\ket\phi^{\otimes N})=0.$

We thus get that
\begin{align}
    \left\Vert \ket{\psi}^{\otimes N}-\ket{\phi}^{\otimes N}\right\Vert &= \sqrt{( 1-a^N)^2+c^2}\nonumber\\
    &=\sqrt{2}\sqrt{1-a^N}\nonumber\\
    &= \sqrt{2}\sqrt{1-\left(1-\frac{\Vert\! \ket\psi-\ket\phi\!\Vert^2 }{2}\right)^N}.
    \end{align} Next, using $(1+x)^N\geq 1+Nx$ when $N\geq 1$ and $x\geq -1,$ we find 
    \begin{align}
        \left\Vert \ket{\psi}^{\otimes N}-\ket{\phi}^{\otimes N}\right\Vert
        &\leq \sqrt{2} \sqrt{1-\left(1-N\frac{\Vert\!\ket\psi-\ket\phi\!\Vert^2}{2}\right)}\nonumber\\
        &= \sqrt{N}\Vert\!\ket\psi-\ket\phi\!\Vert
    \end{align} 
    as claimed.
\end{proof}
\begin{cor}
    Let $P$ be an orthogonal projector on $\mathbb{C}^D$. Then, for all $\ket{\psi}\in\mathbb{C}^D$,
    \begin{equation}\Vert(I-P^{\otimes N})\ket{\psi}^{\otimes N}\!\Vert\leq \sqrt{N}\Vert (I-P)\ket\psi\!\Vert.\end{equation}
\end{cor}
\subsection{Gaussian Tails} In this subsection, we provide a number of bounds on the norms of derivatives, deviation between finite differences and derivatives, and tails of Gaussian states, which we have defined in \cref{defi:gauss}.

First, we bound the normalization constant.
\begin{lemma}\label{lemma:normalization}{\cite[Lemma 5]{watkins2024time}}
    \begin{equation}\frac{1}{\sqrt{\mathcal{N}}}\leq \frac{1}{B}\sqrt\frac{\delta}{\sigma},
    \end{equation}
\end{lemma}
where $B$ is defined in \cref{defi:constants}. 
\begin{proof}
First, note that 
\begin{align}
    \mathcal{N}&= \sum_{j\in\mathbb{Z}_{N_c}}e^{-2\frac{\delta^2}{\sigma^2}\vert j\vert_c^2}\nonumber\\
    &= \sum_{j=0}^{N_c/2-1}e^{-2\frac{\delta^2 j^2}{\sigma^2}} + \sum_{j=N_c/2}^{N_c-1}e^{-2\frac{\delta^2(N_c-j)^2}{\sigma^2}}\nonumber\\
    &\geq 2\sum_{j=0}^{N_c/2-1}e^{-2\frac{\delta^2 j^2}{\sigma^2}}-1.
\end{align}
Next, we note that 
\begin{align}
    \sum_{j=0}^{N_c/2-1}e^{-2\frac{\delta^2 j^2}{\sigma^2}} &\geq \int_0^{N_c/2}\mathrm{d}j\ e^{-2\frac{\delta^2 j^2}{\sigma^2}}\nonumber\\
    &= \frac{\sigma}{\sqrt{2}\delta}\int_0^{\frac{T}{\sqrt{2}\sigma}} \mathrm{d}u\  e^{-u^2} =  \sqrt\frac{\pi}{8}\frac{\sigma}{\delta}\mathrm{erf}\left(\frac{T}{\sqrt{2}\sigma}\right).
\end{align}
Thus,
\begin{equation}
    \mathcal{N} \geq \sqrt\frac{\pi}{2}\frac{\sigma}{\delta}\mathrm{erf}\left(\frac{T}{\sqrt{2}\sigma}\right)-1,
\end{equation}
so 
\begin{align}
    \frac{1}{\mathcal{N}}&\leq  \sqrt\frac{2}{\pi}\frac{\delta}{\sigma}\frac{1}{\mathrm{erf}\left(\frac{ T}{\sqrt{2}\sigma}\right)-\sqrt\frac{2}{\pi}\frac{\delta}{\sigma}}\nonumber\\
    &= \sqrt\frac{2}{\pi}\frac{\delta}{\sigma}\frac{1}{1-\mathrm{erfc}\left(\frac{T}{\sqrt{2}\sigma}\right)-\sqrt\frac{2}{\pi}\frac{\delta}{\sigma}}\nonumber\\
    &\leq \sqrt\frac{2}{\pi}\frac{\delta}{\sigma}\frac{1}{1-\sqrt\frac{2}{\pi}\frac{\sigma}{T}e^{-\frac{T^2}{2\sigma^2}}-\sqrt\frac{2}{\pi}\frac{\delta}{\sigma}},
\end{align}
where we use the fact that 
\begin{equation}\mathrm{erfc}(x)< \frac{e^{-x^2}}{\sqrt{\pi} x}\end{equation} to get the final inequality.
\end{proof}

We next bound the tails.
\begin{lemma}\label{lemma:gausstail}
    Let $\Pi_{\Delta,r}:= \sum_{j : \vert j-r\vert_c> \Delta }\ket{j}\bra{j}.$ Then, with $\ket{\phi_r}$ defined as above in \cref{defi:gauss},
    \begin{equation}\Vert \Pi_{\Delta,r} \ket{\phi_{r}}\Vert \leq \frac{e^{-\frac{\Delta^2\delta^2}{\sigma^2}}}{\sqrt{2}B}\sqrt{\frac{\sigma}{\Delta\delta}}.\end{equation}
\end{lemma}
\begin{proof}
    We have that 
\begin{align}
    \Vert \Pi_{\Delta,r} \ket{\phi_{r}}\Vert^2 &= \frac{1}{\mathcal{N}}\sum_{j : \vert j-r\vert_c> \Delta } e^{-2\frac{\delta^2}{\sigma^2}\vert j-r\vert_c^2}\nonumber\\
    &= \frac{1}{\mathcal{N}}\sum_{j:\vert j\vert_c>\Delta} e^{-2\frac{\delta^2}{\sigma^2}\vert j\vert_c^2}\nonumber\\
    &= \frac{1}{\mathcal{N}}\left[\sum_{j=\Delta+1 }^{N_c/2-1}e^{-2 j^2\delta^2/\sigma^2}+\sum_{j=N_c/2}^{N_c-\Delta-1}e^{-2(N_c-j)^2\delta^2/\sigma^2}\right]\nonumber\\
    &\leq \frac{2}{\mathcal{N}}\sum_{j=\Delta +1}^{N_c/2}e^{-2j^2\frac{\delta^2}{\sigma^2}}.\end{align}

Now, we can bound this sum from above by an improper integral from $\Delta$ to $\infty,$ getting
\begin{align}
    \Vert \Pi_{\Delta,r}\ket{\phi_r}\Vert^2\leq \frac{2}{\mathcal{N}} \int_{\Delta}^\infty \mathrm{d}je^{-2\frac{\delta^2}{\sigma^2}j^2}.
\end{align}
Now, we can use the fact that $\int_\beta^\infty \mathrm{d}xe^{-\alpha^2x^2}\leq \frac{e^{-\alpha^2\beta^2}}{2\alpha^2\beta}$
to say that 
\begin{align}
    \Vert \Pi_{\Delta,r}\ket{\phi_r}\Vert^2\leq \frac{1}{2\mathcal{N}}\frac{\sigma^2}{\delta^2\Delta}e^{-2\frac{\delta^2\Delta^2}{\sigma^2}}.
\end{align}
Then, by \cref{lemma:normalization}, we have  \begin{equation}\Vert \Pi_{\Delta,r}\ket{\phi_r}\Vert^2\leq \frac{1}{2B^2}\frac{\sigma}{\delta\Delta}e^{-2\frac{\delta^2\Delta^2}{\sigma^2}},\end{equation} arriving at our result.
\end{proof}

Before proving more results about Gaussian states, we introduce a technical lemma bounding the difference between first and second finite differences and first and second derivatives of functions.
\begin{lemma}\label{lemma:finitedifferenceterms}
    Let $f\in C^{4}([t_0-h,t_0+h]).$ Then 
    \begin{equation}\left\vert [D_hf](t_0)-[\partial_tf](t_0)\right\vert\leq \frac{h^2}{6}\max_{\tau\in [t_0-h,t_0+h]}\left\vert f'''(\tau)\right\vert \end{equation}
    and 
    \begin{equation}\left\vert [D^2_hf](t_0)-[\partial^2_tf](t_0)\right\vert\leq \frac{h^2}{3}\max_{\tau\in [t_0-2h,t_0+2h]}\left\vert f'''' (\tau)\right\vert.  \end{equation}
\end{lemma}
\begin{proof}[Proof (adapted from Ref.~\cite{mathseapproximation}).]
    Note that $f(t_0\pm h) = f(t_0)\pm f'(t_0)h + \frac{f''(t_0)}{2}h^2 + R_2(t_0, \pm h),$ where \begin{equation}R_k(t_0, \delta):= \int_{t_0}^{t_0+\delta}\mathrm{d}t \frac{f^{(k+1)}(t)}{k!}(t-t_0)^k.\end{equation} Thus,
    \begin{align}
        [D_{h}f](t_0)&=\frac{(f(t_0)+ f'(t_0)h + \frac{f''(t_0)}{2}h^2 + R_2(t_0,  h))-(f(t_0)- f'(t_0)h + \frac{f''(t_0)}{2}h^2 + R_2(t_0,  -h))}{2h}\nonumber\\
        &= f'(t_0) + \frac{R_2(t_0,h)-R_2(t_0,-h)}{2h}.
    \end{align}
We can now bound the remainder term
\begin{align}
    \frac{R_2(t_0,h)-R_2(t_0,-h)}{2h}&= \frac{1}{2h}\left[\int_{t_0}^{t_0+h}\mathrm{d}t \frac{f'''(t)}{2!}(t-t_0)^2-\int_{t_0}^{t_0-h}\mathrm{d}t \frac{f'''(t)}{2!}(t-t_0)^2\right]\nonumber\\
    &= \frac{1}{4h}\int_{t_0-h}^{t_0+h}\mathrm{d}t f'''(t)(t-t_0)^2,
\end{align}
giving
\begin{align}
   \left\vert \frac{R_2(t_0,h)-R_2(t_0,-h)}{2h}\right\vert &\leq \frac{1}{4h}\max_{t\in[t_0-h,t_0+h]}\vert f'''(t)\vert \times \left[ \frac{(t-t_0)^3}{3}\right]_{t_0-h}^{x+h}\nonumber\\
   &= \frac{h^2}{6}\max_{t\in[t_0-h,t_0+h]}\vert f'''(t)\vert. 
\end{align}
We can perform a similar analysis for the second-order finite difference:
    \begin{align}
        [D^2_{h}f](t_0)&=\frac{1}{(2h)^2}\Bigg[\left(f(t_0)+ f'(t_0)(2h) + \frac{f''(t_0)}{2}(2h)^2 + \frac{f'''(t_0)}{6}(2h)^3+R_3(t_0,  2h)\right)\nonumber\\&\qquad\qquad\qquad-2f(t_0)+\left(f(t_0)- f'(t_0)(2h) + \frac{f''(t_0)}{2}(2h)^2-\frac{f'''(t_0)}{6}(2h)^3 + R_3(t_0,  -2h)\right)\Bigg]\nonumber\\
        &= f''(t_0) + \frac{R_3(t_0,2h)+R_3(t_0,-2h)}{(2h)^2}.
    \end{align}
Bounding the remainders, we get
    \begin{align}
        \frac{R_3(t_0,2h)+R_3(t_0,-2h)}{(2h)^2} &= \frac{1}{(2h)^2} \left[\int_{t_0}^{t_0+2h}\mathrm{d}t \frac{f''''(t)}{6}(t-t_0)^3 + \int_{t_0}^{t_0-2h}\mathrm{d}t \frac{f''''(t)}{6}(t-t_0)^3\right]
\end{align}
so
\begin{align}
        \left\vert\frac{R_3(t_0,2h)+R_3(t_0,-2h)}{(2h)^2}\right\vert &\leq \frac{1}{6(2h)^2} \max_{t\in [t_0-2h,t_0+2h]}\left\vert f''''(t)\right\vert \times\frac{(2h)^4}{2}\nonumber\\
       &= \frac{h^2}{3}\max_{t\in [t_0-2h,t_0+2h]}\left\vert f''''(t)\right\vert.
    \end{align}
\end{proof}

We now bound the difference of the finite difference and derivative \emph{states} of the Gaussians.
\begin{lemma}\label{lemma:finitedifferenceerror}
    Let $\phi_j(t):=\frac{1}{\sqrt{\mathcal{N}}}\exp\left[{-\frac{\delta^2}{\sigma^2}\left\vert\frac{t}{\delta}-j\right\vert_c^2}\right].$ Then, if $\ket{D^{(k)}_{\delta}\phi_j}:= \sum_{l}[D_{\delta}^{(k)}\phi_j](l\delta)\ket{l}$, with $D_\delta$ defined as above in \cref{defi:partialdiff}, and $\ket{\partial_t^{(k)}\phi} := \sum_{l\notin \{j+N_c/2-(k-1),\dots,j+N_c/2+(k-1)\}}\phi_j^{(k)}(l\delta)\ket{l},$ \begin{equation}\Vert\ket{D_{\delta}\phi_j}-\ket{\partial_t\phi_j}\Vert \leq \sqrt2\frac{\delta^2}{\sigma^3}\frac{A}{B}\end{equation} and
    \begin{equation}\Vert\ket{D^2_{\delta}\phi_j}-\ket{\partial^2_t\phi_j}\Vert \leq 8\frac{\delta^2}{\sigma^4}\frac{A}{B}+\frac{1}{B}\frac{e^{-\frac{(N_c-6)^2\delta^2}{4\sigma^2}} }{\sqrt{\sigma \delta^3}},\end{equation}
    where $A,B$ are defined as in \cref{defi:constants}.
\end{lemma}
\begin{proof}
Let $\phi_j(t):= \frac{1}{\sqrt{\mathcal{N}}}e^{-\frac{\delta^2}{\sigma^2}\vert \frac{t}{\delta}-j\vert_c^2.}$ We then have that 
    \begin{align}
        \ket{D_{\delta}\phi_j} &= \sum_{k\in \mathbb{Z}_{N_c}} [D_\delta\phi_j](k\delta)\ket{k},\\
        \ket{\partial_t\phi_j}&= \sum_{k\in \mathbb{Z}_{N_c}\backslash\{ j+N_c/2\}} \dot{\phi}_j(k\delta)\ket{k}. \end{align}

Then, by \cref{lemma:finitedifferenceterms},
    \begin{align}
        \Vert (I-\ket{j+N_c/2}\bra{j+N_c/2})(\ket{D_{\delta}\phi_j}-\ket{\partial_t\phi_j})\Vert^2 &\leq \frac{\delta^4}{36}\sum_{k\in\mathbb{Z}_{N_c}\backslash\{j+N_c/2\}} \max_{t\in[(k-1)\delta,(k+1)\delta]}\left\vert \dddot{\phi}_j(t)\right\vert^2.
    \end{align}
    Now, we can put in a bound on $\dddot{\phi}_j(t).$ We have that 
    \begin{align}
        \left\vert\dddot{\phi}_j(t)\right\vert&=\frac{4}{\sigma^3\sqrt{\mathcal{N}}}  \left\vert 3 \frac{\delta\vert \frac{t}{\delta}-j\vert_c}{\sigma} - 2 \frac{\delta^3\vert \frac{t}{\delta}-j\vert_c^3}{\sigma^3}\right\vert e^{-\frac{\delta^2}{\sigma^2}\vert \frac{t}{\delta}- j\vert_c^2}.
    \end{align}
    Noting that 
    \begin{equation}\vert3x-2x^3\vert < \sqrt{2} e^{x^2/2},\end{equation} we have
    \begin{equation}\vert \dddot{\phi}_j(t)\vert < \frac{4\sqrt{2}}{\sigma^3\sqrt{\mathcal{N}}}e^{-\frac{\delta^2}{2\sigma^2}\vert \frac{t}{\delta}- j\vert_c^2}.\end{equation} 
    Thus, we can insert this bound to get 
    \begin{align}
         \Vert (I-\ket{j+N_c/2}\bra{j+N_c/2})(\ket{D_{\delta}\phi_j}-\ket{\partial_t\phi_j})\Vert^2 &\leq \frac{\delta^4}{\sigma^6\mathcal{N}}\sum_{k\in\mathbb{Z}_{N_c}\backslash\{j+N_c/2\}} \max_{t\in[(k-1)\delta,(k+1)\delta]}e^{-\frac{\delta^2}{\sigma^2}\vert \frac{t}{\delta}-j\vert_c^2}\nonumber\\
         &=\frac{\delta^4}{\sigma^6\mathcal{N}}\Bigg[1+\sum_{k=j+1}^{j+N_c/2-1}e^{-\frac{(k-1-j)^2\delta^2}{\sigma^2}}\nonumber\\&\qquad\qquad+\sum_{k=j+N_c/2+1}^{j+N_c-1}e^{-\frac{(N_c-(k+1-j))^2\delta^2}{\sigma^2}}\Bigg].
         \end{align}
         Once again bounding the sum by an integral, we find 
         \begin{align}
         \Vert (I-\ket{j+N_c/2}\bra{j+N_c/2})(\ket{D_{\delta}\phi_j}-\ket{\partial_t\phi_j})\Vert^2&\leq \frac{\delta^4}{\sigma^6\mathcal{N}}\left[3+\frac{2\sigma}{\delta}\int_{0}^\infty\mathrm{d}u\ e^{-u^2}\right]\nonumber\\
         &\leq 2\frac{A^2}{B^2}\frac{\delta^4}{\sigma^6}.
    \end{align}
    Next, we consider the  $\ket{j+N_c/2}$ term. We note that 
    \begin{align}
         \ket{j+N_c/2}\bra{j+N_c/2}(\ket{D_{\delta}\phi_j}-\ket{\partial_t\phi_j}) &= (\ket{j+N_c/2}\bra{j+N_c/2})\ket{D_{\delta}\phi_j}\nonumber\\
         &= \frac{\phi_j((j+N_c/2+1)\delta)-\phi_j((j+N_c/2-1)\delta)}{2\delta}\ket{j+N_c/2}\nonumber\\
         &= 0.
    \end{align} Thus, we have 
    \begin{equation}\Vert \ket{D_{\delta}\phi_j}-\ket{\partial_t\phi_j}\Vert=\Vert (I-\ket{j+N_c/2}\bra{j+N_c/2})(\ket{D_{\delta}\phi_j}-\ket{\partial_t\phi_j})\Vert,\end{equation} establishing the first part of our lemma.
    
    Next, consider $\ket{D_{\delta}^2\phi_j}.$ Let \begin{equation}\varpi:= I-\ket{j+N_c/2}\bra{j+N_c/2}-\ket{j+N_c/2+1}\bra{j+N_c/2+1}-\ket{j+N_c/2-1}\bra{j+N_c/2-1}.\end{equation} We have
    \begin{align}
        \Vert \varpi (\ket{D^2_{\delta}\phi_j}-\ket{\partial^2_t\phi_j})\Vert^2&\leq \frac{\delta^4}{9}\sum_{k-(j+N_c/2)\notin \{-1,0,1\}}\max_{t\in[(k-2)\delta,(k+2)\delta]}\vert\ddddot{\phi}_j(t)\vert^2. 
    \end{align}
and
\begin{equation}\ddddot{\phi}_j(t)=  \frac{4 e^{-\frac{\delta^2}{\sigma^2}\vert \frac{t}{\delta}-j\vert_c^2}}{\sigma^4\sqrt{\mathcal{N}}} \left\vert 3  - 12  \left(\frac{\delta\vert \frac{t}{\delta}-j\vert_c}{\sigma}\right)^2 + 4 \left(\frac{\delta\vert \frac{t}{\delta}-j\vert_c}{\sigma}\right)^4\right\vert.\end{equation} As
\begin{equation}\left\vert 3-12x^2+4x^4\right\vert<4e^{x^2/2},\end{equation} we have 
\begin{equation}\vert \ddddot{\phi}_j(t)\vert < \frac{16e^{-\frac{\delta^2}{2\sigma^2}\vert \frac{t}{\delta}-j\vert_c^2}}{\sigma^4\sqrt{\mathcal{N}}}.\end{equation} Plugging this in similarly as above, we get
\begin{align}
    \Vert \varpi (\ket{D_{\delta}^2\phi_j}-\ket{\partial_t^2\phi_j})\Vert^2 &\leq \frac{256}{9\mathcal{N}}\left(\frac{\delta}{\sigma^2}\right)^4 \sum_{k-(j+N_c/2)\notin\{-1,0,1\}}\max_{t\in [(k-2)\delta,(k+2)\delta]}e^{-\frac{\delta^2}{\sigma^2}\vert \frac{t}{\delta}-j\vert_c^2}\nonumber\\
    &= \frac{256}{9\mathcal{N}}\left(\frac{\delta}{\sigma^2}\right)^4\left[5+\sum_{k=j+3}^{j+N_c/2-2}e^{-\frac{(k-2-j)^2\delta^2}{\sigma^2}} + \sum_{k=j+N_c/2+2}^{j+N_c-3}e^{-\frac{(N_c-(k+2-j))^2\delta^2}{\sigma^2}}\right].
    \end{align}
Reindexing and combining the sums, we get that
\begin{align}
   \Vert \varpi (\ket{D_{\delta}^2\phi_j}-\ket{\partial_t^2\phi_j})\Vert^2 &\leq \frac{512}{9\mathcal{N}}\left(\frac{\delta}{\sigma^2}\right)^4\left[\frac{5}{2}+\sum_{k=1}^{N_c/2-4}e^{-\frac{k^2\delta^2}{\sigma^2}}\right]\nonumber\\
    &\leq \frac{512}{9\mathcal{N}}\left(\frac{\delta}{\sigma^2}\right)^4\left[\frac{5}{2}+\int_{0}^{\infty}\mathrm{d}ke^{-\frac{k^2\delta^2}{\sigma^2}}\right]\nonumber\\
    &= \frac{512}{9\mathcal{N}}\left(\frac{\delta}{\sigma^2}\right)^4\left[\frac{5}{2}+\frac{\sigma\sqrt{\pi }}{2\delta}\right]\nonumber\\
    &\leq \frac{256\sqrt\pi}{9\mathcal{N}}\frac{\delta^{3}}{\sigma^7} \left[1 + \frac{5}{\sqrt\pi}\frac{\delta}{\sigma}\right]\nonumber\\
    &\leq \left(\frac{16\pi^{1/4}}{3}\frac{A}{B}\frac{\delta^2}{\sigma^4}\right)^2  \leq \left(8\frac{A}{B}\frac{\delta^2}{\sigma^4}\right)^2.
\end{align}

Now, we bound the $\ket{j+N_c/2-1},\ket{j+N_c/2},$ and $\ket{j+N_c/2+1}$ terms. For the $\ket{j+N_c/2\pm 1}$ terms, we have 
\begin{align}
    \vert \bra{j+N_c/2\pm 1}(\ket{D_{\delta}^2\phi_j}-\ket{\partial_t^2\phi_j})\vert &= \left\vert\frac{\phi_j((j+N_c/2\pm3)\delta)-2\phi_j((j+N_c/2\pm 1)\delta)+\phi_j((j+N_c/2\mp 1)\delta)}{(2\delta)^2}\right\vert\nonumber\\
    &\leq \frac{e^{-\frac{(N_c/2-3)^2\delta^2}{\sigma^2}}-e^{-\frac{(N_c/2-1)^2\delta^2}{\sigma^2}}}{\sqrt{\mathcal{N}}(2\delta)^2}\nonumber\\
    &\leq \frac{e^{-\frac{(N_c/2-3)^2\delta^2}{\sigma^2}}}{4\delta^2\sqrt{\mathcal{N}}}.
\end{align}
We then have that the $\ket{j+N_c/2}$ term is 
\begin{align}
    \vert \bra{j+N_c/2}(\ket{D_{\delta}^2\phi_j}-\ket{\partial_t^2\phi_j})\vert &= \left\vert\frac{\phi_j((j+N_c/2+2)\delta)-2\phi_j((j+N_c/2)\delta)+\phi_j((j+N_c/2-2)\delta)}{(2\delta)^2}\right\vert\nonumber\\
    &= \frac{1}{2\delta^2\sqrt{\mathcal{N}}} \left(e^{-\frac{(N_c/2-2)^2\delta^2}{\sigma^2}}-e^{-\frac{(N_c/2)^2\delta^2}{\sigma^2}}\right)\nonumber\\
    &\leq \frac{e^{-\frac{(N_c/2-2)^2\delta^2}{\sigma^2}}}{2\delta^2\sqrt{\mathcal{N}}}.
\end{align}
Thus, by the triangle inequality, we have 
\begin{align}\Vert \ket{D_{\delta}^2\phi_j}-\ket{\partial_t^2\phi_j}\Vert &\leq 8\frac{A}{B}\frac{\delta^2}{\sigma^4}
+\frac{1}{2\delta^2\sqrt{\mathcal{N}}}\left(\frac{e^{-\frac{(N_c/2-3)^2\delta^2}{\sigma^2}}}{2} +e^{-\frac{(N_c/2-2)^2\delta^2}{\sigma^2}}\right)\nonumber\\
&\leq 8\frac{A}{B}\frac{\delta^2}{\sigma^4}+\frac{e^{-\frac{(N_c/2-3)^2\delta^2}{\sigma^2}} }{\delta^2\sqrt{\mathcal{N}}}\nonumber\\
&\leq \frac{16}{3}\frac{A}{B}\frac{\delta^2}{\sigma^4}
\frac{e^{-\frac{(N_c/2-3)^2\delta^2}{\sigma^2}} }{B\sqrt{\sigma \delta^3}
}.
\end{align}
\end{proof}
Next, we bound the norm of the second Gaussian derivative state. 
\begin{lemma}\label{lemma:doublepartial}
    Let $\ket{\partial_t^2\phi_j}$ be defined as above in \cref{lemma:finitedifferenceerror}. Then 
    \begin{equation}\Vert \ket{\partial_t^2\phi_j}\Vert \leq 4\frac{A}{B}\frac{1}{\sigma^2}
    .\end{equation}
\end{lemma}
\begin{proof}
    Away from $j\delta +T,$ we have
        \begin{equation}\ddot{\phi}_j(t)= \frac{1}{\sigma^2\sqrt{\mathcal{N}}}\left(4\frac{\delta^2}{\sigma^2}\left\vert \frac{t}{\delta}-j\right\vert_c^2-2\right)e^{-\frac{\delta^2}{\sigma^2}\vert \frac{t}{\delta}-j\vert_c^2}.\end{equation} As $\vert 4x^2-2\vert\leq 2\sqrt{2}e^{-x^2}, $ we have 
        \begin{equation}\vert\ddot{\phi}_j(t)\vert\leq \frac{2\sqrt{2}}{\sigma^2\sqrt{\mathcal{N}}}e^{-\frac{\delta^2}{2\sigma^2}\vert \frac{t}{\delta}-j\vert_c^2}.\end{equation} Thus, 
        \begin{align}
            \Vert \ket{\partial_t^2\phi_j}\Vert^2 &\leq \frac{8}{\sigma^4\mathcal{N}}\sum_{k\neq j+N_c/2}e^{-\frac{\vert t-j\delta \vert_c^2}{\sigma^2}}\nonumber\\
            &\leq\frac{8}{\sigma^4\mathcal{N}}\left[ 1 + 2\sum_{k=1}^{\infty}e^{-\frac{k^2\delta^2}{\sigma^2}}\right]\nonumber\\
            &\leq \frac{8\sqrt{\pi}}{\sigma^4\mathcal{N}}\frac{\sigma}{\delta }\left[1 + \frac{\delta}{\sigma\sqrt{\pi}}\right]\nonumber\\
            &\leq \frac{8\sqrt\pi}{B^2\sigma^4}\left[1 + \frac{\delta}{\sigma\sqrt{\pi}}\right].
        \end{align}
\end{proof}

Finally, we bound the error of a first-order Taylor approximation of shifting the Gaussians.

\begin{lemma}\label{lemma:firstordershift}
\begin{equation}\Vert \ket{\phi_j}-\delta \ket{\partial_t\phi_j}-\ket{\phi_{j+1}}\Vert\leq 2\frac{A}{B}\left(\frac{\delta}{\sigma}\right)^2+\sqrt{\frac{\delta}{\sigma}}\frac{e^{-\frac{(N_c-2)^2\delta^2}{4\sigma^2}}}{B}.\end{equation}
\end{lemma}
\begin{proof}   
    
    We first note that
    \begin{align}
        \ket{\phi_j}-\delta \ket{\partial_t \phi_j}  &= \sum_{k\neq j+N_c/2} (\phi_j(k\delta)-\delta \dot{\phi}_j(k\delta) )\ket{k}+\phi_j((j+N_c/2)\delta)\ket{j+N_c/2},\\
        \phi_{j}(k\delta) &= \phi_{j+1}((k+1)\delta),\\
        \phi_j(k\delta)-\delta\dot{\phi}_j(k\delta)- \phi_{j+1}(k\delta) &=\phi_j(k\delta)-\delta\dot{\phi}_j(k\delta)-\phi_j((k-1)\delta)\nonumber\\ &=\int_{k\delta}^{(k-1)\delta}\mathrm{d}\tau\ddot{\phi}_j(\tau)(\tau-k\delta).\end{align} 
        Thus, we have
        \begin{align}
        \ket{\phi_j}-\delta\ket{\partial_t\phi_j}-\ket{\phi_{j+1}} &= \sum_{k\neq j+N_c/2} \left[\int_{k\delta}^{(k-1)\delta}\mathrm{d}\tau \ddot{\phi}_j(\tau)(\tau-k\delta)\right]\ket{k}\nonumber\\
        &\qquad+  (\phi_j((j+N_c/2)\delta)-\phi_{j+1}((j+N_c/2)\delta))\ket{j+N_c/2}.\end{align} Bounding the $\ket{j+N_c/2}$ component, we have that
        \begin{align}
        \left\vert \phi_j((j+N_c/2)\delta)-\phi_{j+1}((j+N_c/2)\delta)\right\vert &=  \left\vert \phi_j((j+N_c/2)\delta)-\phi_{j}((j+N_c/2-1)\delta)\right\vert\nonumber\\
        &= \left\vert\frac{e^{-\frac{N_c^2\delta^2}{4\sigma^2}}-e^{-\frac{(N_c/2-1)^2\delta^2}{\sigma^2}}}{\sqrt{\mathcal{N}}}\right\vert\nonumber\\
        &\leq \sqrt{\frac{\delta}{\sigma}}\frac{e^{-\frac{(N_c-2)^2\delta^2}{4\sigma^2}}}{B}
        .\end{align}
        Next, we note that, away from $j\delta +T,$
        \begin{equation}\ddot{\phi}_j(t)= \frac{1}{\sigma^2\sqrt{\mathcal{N}}}\left(4\frac{\delta^2}{\sigma^2}\left\vert \frac{t}{\delta}-j\right\vert_c^2-2\right)e^{-\frac{\delta^2}{\sigma^2}\vert \frac{t}{\delta}-j\vert_c^2}.\end{equation} As $\vert 4x^2-2\vert\leq 2\sqrt{2}e^{-x^2}, $ we have 
        \begin{equation}\vert\ddot{\phi}_j(t)\vert\leq \frac{2\sqrt{2}}{\sigma^2\sqrt{\mathcal{N}}}e^{-\frac{\delta^2}{2\sigma^2}\vert \frac{t}{\delta}-j\vert_c^2}.\end{equation} 
        Thus, we find
        \begin{align}
            \left\vert\int_{(k-1)\delta}^{k\delta}\mathrm{d}\tau \ddot{\phi}_j(\tau)(\tau-k\delta)\right\vert &\leq \frac{\delta^2}{2}\max_{\tau\in[(k-1)\delta,k\delta]}\left\vert \ddot\phi_j(\tau)\right\vert \nonumber\\
            &\leq \frac{\sqrt{2}}{\sqrt{\mathcal{N}}}\left(\frac{\delta}{\sigma}\right)^2 \max_{\tau\in[(k-1)\delta,k\delta]}e^{-\frac{\delta^2}{2\sigma^2}\vert \frac{t}{\delta}-j\vert_c^2}.
        \end{align}
We can use this to take the norm of the vector, getting 
        \begin{align}
        \left\Vert \sum_{k\neq j+N_c/2} \int_{(k-1)\delta}^{k\delta}\mathrm{d}\tau\ddot{\phi}_j(\tau)(\tau-k\delta)\ket{k}\right\Vert^2&\leq \frac{2}{\mathcal{N}}\left(\frac{\delta}{\sigma}\right)^4\sum_{k\neq j+N_c/2}\max_{\tau\in [(k-1)\delta,k\delta]}e^{-\frac{\delta^2}{\sigma^2}\vert\frac{t}{\delta}-j\vert_c^2}\nonumber\\
        &= \frac{2}{\mathcal{N}}\left(\frac{\delta}{\sigma}\right)^4\left[\sum_{k=0}^{N_c/2-2} e^{-\frac{k^2\delta^2}{\sigma^2}}+\sum_{k=0}^{N_c/2-1} e^{-\frac{k^2\delta^2}{\sigma^2}}\right].\end{align}
        Once again using the integral bound, we get 
        \begin{align}
        \left\Vert \sum_{k\neq j+N_c/2} \int_{(k-1)\delta}^{k\delta}\mathrm{d}\tau\ddot{\phi}_j(\tau)(\tau-k\delta)\ket{k}\right\Vert^2&\leq  \frac{4(\delta/\sigma)^4}{\mathcal{N}}\left[1 + \int_0^{\infty} \mathrm{d}ke^{-\frac{k^2\delta^2}{\sigma^2}}\right]\nonumber\\
        &= \frac{4(\delta/\sigma)^4}{\mathcal{N}}\left[1 + \frac{\sigma\sqrt\pi}{2\delta}\right]\nonumber\\
        &\leq 2\sqrt{\pi}\left(\frac{\delta}{\sigma}\right)^4\frac{A^2}{B^2}.
    \end{align}
    Thus, by the triangle inequality, we get
    \begin{equation}\Vert \ket{\phi_j}-\delta\ket{\partial_t\phi_j}-\ket{\phi_{j+1}}\Vert \leq 2\frac{A}{B}\left(\frac{\delta}{\sigma}\right)^2
    +\sqrt{\frac{\delta}{\sigma}}\frac{e^{-\frac{(N_c-2)^2\delta^2}{4\sigma^2}}}{B}\end{equation}
    as claimed.
\end{proof}

\subsection{Trotter Error Bounds}

In this subsection, we state two basic lemmas that we use in the full error analysis. 
First, we use a bound on the first-order Trotter product error. 
\begin{lemma}[First-Order Trotter Product Formula {{\cite[Eq. 2.3]{Huyghebaert1990}
}}]\label{lemma:firstordertrotter}
     Let $A,B$ be Hermitian. Then 
\begin{equation}\Vert e^{-i(A+B)t}-e^{-iAt}e^{-iBt}\Vert \leq \frac{t^2}{2}\Vert [A,B]\Vert.\end{equation}
\end{lemma}

We also use {the following standard error bound relating the evolution of a time-dependent Hamiltonian and a piecewise-constant approximation of it, which we prove here for completeness.}
\begin{lemma}[{Piecewise-Constant Hamiltonian Approximation Error}]
\label{lemma:discretizationapprox}
    Let $U(t)$ be the unitary generated by 
\begin{equation}i\dot{U}:=H(t)\end{equation} with the initial condition $U(0)=I,$ and let $V_T:= \prod_{k=0}^{N-1} e^{-i \frac{T}{N}H(k T/N)}.$ Then 
\begin{equation}\Vert U(T)-V_T\Vert \leq \frac{T^2}{N}\max_{t}\Vert \dot{H}(t)\Vert. \end{equation}
\end{lemma}
\begin{proof}
    Note that $V_T$ is generated by the time-dependent Hamiltonian \begin{equation}\tilde{H}(t):= H\left(\frac{T}{N}\left\lfloor \frac{Nt}{T}\right\rfloor\right)\end{equation} acting for time $T$. Let $V_T(t)$ be the unitary generated by $\tilde{H}(t)$ from time $0$ to $t.$ Then $V_T(T)=V_T.$ We have  
\begin{align}
    U(T)-V_T(T) &= U(T)(I-U^\dagger(T)V_T(t)) \nonumber\\
    &= iU(T)\int_0^T\mathrm{d}t U^\dagger(t)(H(t)-\tilde{H}(t))V_T(t).\end{align}
   Therefore,
    \begin{align}
    \left\Vert U(T)-V_T(T)\right\Vert &= \left\Vert \int_0^{T} \mathrm{d}t U^\dagger(t)(H(t)-\tilde{H}(t))V_T(t)\right\Vert\nonumber\\
    &= \left\Vert\int_0^T\mathrm{d}tU^\dagger(t) \left(H(t) - H\left(\frac{T}{N}\left\lfloor \frac{Nt}{T}\right\rfloor\right)\right)V_T(t)\right\Vert\nonumber\\
    &=\left\Vert\int_0^T\mathrm{d}tU^\dagger(t) \left[\int_{\frac{T}{N}\left\lfloor \frac{Nt}{T}\right\rfloor}^{t}\mathrm{d}\tau \dot{H}(\tau)\right]V_T(t)\right\Vert\nonumber\\
    &\leq \int_{0}^T\mathrm{d}t \left\Vert\int_{\frac{T}{N}\left\lfloor \frac{Nt}{T}\right\rfloor}^{t}\mathrm{d}\tau \dot{H}(\tau)\right\Vert\nonumber\\
    &\leq \int_{0}^T\mathrm{d}t\int_{\frac{T}{N}\left\lfloor \frac{Nt}{T}\right\rfloor}^{t}\mathrm{d}\tau\Vert \dot{H}(\tau)\Vert \leq \frac{T^2}{N}\max_{t}\Vert \dot{H}(t)\Vert.\end{align}\end{proof}
\section{Proofs}\label{section:proofs}
In this section, we compile the rigorous statements and proofs of the results we explain intuitively in the main text.

First, we provide a bound on the commutator $\Vert [\boldsymbol{\Delta},C(H)]\Vert.$
\begin{lemma}\label{lemma:commutator} Let $\boldsymbol{\Delta},C(H)$ be defined as in \cref{defi:suppWWRL}. Then
    \begin{align}[\boldsymbol{\Delta},C(H)]=-i\mathrm{Re}\Bigg[\sum_{\mathbf{k}\in\mathbb{Z}_{N_c}^N}\Bigg(&\sum_{e} \left(\frac{H_{e}((k_{\rho(e)}+1)\delta)-H_{e}({k_{\rho(e)}}\delta)}{\delta }\right)\otimes U_{\rho(e)} \nonumber\\&\quad{ +\sum_{v}\left(\frac{H_v((k_v+1)\delta)-H_v(k_v\delta)}{\delta}\right)\otimes U_v}\Bigg)(I\otimes\ket{\mathbf{k}}\!\!\bra{\mathbf{k}})\Bigg],\end{align} and thus $\Vert [\boldsymbol{\Delta},C(H)]\Vert \leq h_1.$
\end{lemma}

\begin{proof}[Proof (Adapted from \cite{watkins2024time})]
 First, we note that 
    \begin{align}     [\boldsymbol{\Delta}, C(H)] :&= \left[\sum_{v}\Delta_v, \sum_{\mathbf{k}}\left(\sum_{e} H_{e}(k_{\rho(e)}\delta)+\sum_{v}H_v(k_v\delta)\right)\otimes \ket{\mathbf{k}}\bra{\mathbf{k}}\right]\nonumber\\
    &= \frac{i}{2\delta}\sum_{v}\left[(U_{v}-U^\dagger_{v}), \sum_{\mathbf{k}}\left(\sum_{e} H_{e}(k_{\rho(e)}\delta)+\sum_{v}H_v(k_v\delta)\right)\otimes \ket{\mathbf{k}}\bra{\mathbf{k}}\right]\nonumber\\
    &= \frac{i}{\delta}\mathrm{Re}\left[\sum_{v} U_v, \sum_{\mathbf{k}}\left(\sum_{e}H_{e}(k_{\rho(e)}\delta)+\sum_{v}H_v(k_v\delta)\right)\otimes \ket{\mathbf{k}}\bra{\mathbf{k}}\right],\end{align}
where to get to the last line we use that $[X,Y]^\dagger = -[X^\dagger, Y^\dagger].$ We can then directly apply the $U_v$s, getting
    \begin{align}
    [\boldsymbol{\Delta},C(H)]&= \frac{i}{\delta} \mathrm{Re} \left(\sum_{v{ '}}\sum_{\mathbf{k}}\left(\sum_{e}H_{e}(k_{\rho(e)}\delta){ +\sum_{v}H_v(k_v\delta)}\right)\otimes \left(\ket{\mathbf{k}+\mathbf{e}_{v{ '}}}\bra{\mathbf{k}}-\ket{\mathbf{k}}\bra{\mathbf{k}-\mathbf{e}_{v{ '}}}\right)\right)\nonumber\\
    &= \frac{i}{\delta} \mathrm{Re} \left(\sum_{v{ '}}U_{v{ '}}\sum_{\mathbf{k}}\left(\sum_{e}H_e(k_{\rho(e)}\delta){ +\sum_{v}H_v(k_v\delta)}\right)\otimes \left(\ket{\mathbf{k}}\bra{\mathbf{k}}-\ket{\mathbf{k}-\mathbf{e}_{v{ '}}}\bra{\mathbf{k}-\mathbf{e}_{v{ '}}}\right)\right)\nonumber\\
    &= -i\mathrm{Re}\Bigg[\sum_{\mathbf{k}}\sum_{e} \left(\frac{H_{e}((k_{\rho(e)}+1)\delta)-H_{e}(k_{\rho(e)}\delta)}{\delta}\right)\otimes U_{\rho(e)}\ket{\mathbf{k}}\bra{\mathbf{k}}\nonumber\\&\qquad \qquad{ + \sum_{\mathbf{k}}\sum_v\left(\frac{H_v((k_v+1)\delta)-H_v(k_v\delta)}{\delta}\right)\otimes U_v\ket{\mathbf{k}}\bra{\mathbf{k}}} \Bigg].\end{align}
    Now, putting on the norms, we can switch the sums and use the triangle inequality to find
    \begin{align}
    \Vert [C(H),\Delta]\Vert
    &\leq \sum_{e} \left\Vert \sum_{\mathbf{k}}\left(\frac{H_{e}((k_{\rho(e)}+1)\delta)-H_{e}(k_{\rho(e)}\delta)}{\delta}\right)\otimes U_{\rho(e)}\ket{\mathbf{k}}\bra{\mathbf{k}}\right\Vert\nonumber\\
    &\qquad{ + \sum_{v}\left\Vert \sum_{\mathbf{k}}\left(\frac{H_v((k_v+1)\delta)-H_v(k_v\delta)}{\delta}\right)\otimes U_v\ket{\mathbf{k}}\bra{\mathbf{k}}\right\Vert }\nonumber\\
    &= \sum_{e} \left\Vert \sum_{\mathbf{k}}\left(\frac{H_{e}((k_{\rho(e)}+1)\delta)-H_{e}(k_{\rho(e)}\delta)}{\delta}\right)\otimes \ket{\mathbf{k}}\bra{\mathbf{k}}\right\Vert\nonumber\\
    &\qquad{ +\sum_{v}\left\Vert\sum_{\mathbf{k}}\left(\frac{H_v((k_v+1)\delta)-H_v(k_v\delta)}{\delta}\right)\otimes\ket{\mathbf{k}}\bra{\mathbf{k}}\right\Vert }\nonumber\\
    &= \sum_{e} \max_{k_{\rho(e)}}\left\Vert\frac{H_{e}((k_{\rho(e)}+1)\delta)-H_{e}(k_{\rho(e)}\delta)}{\delta}\right\Vert{ +\sum_{v} \max_{k_{v}}\left\Vert\frac{H_{v}((k_{v}+1)\delta)-H_{v}(k_{v}\delta)}{\delta}\right\Vert}\nonumber\\&\leq\sum_{{ x\in E\cup V}}\max_{t\in[0,T]}\left\Vert\dot{H}_{{ x}}(t)\right\Vert=h_1.
    \end{align}
\end{proof}

\begin{lemma}\label{lemma:translation_error}
 Let $\mathbf{\Delta},$ $\ket{\Phi_{\mathbf{j}}}$ be defined as above in Definitions \ref{defi:suppWWRL} and \ref{defi:gauss} respectively. Then
    \begin{align}\left\Vert e^{- im\delta  \boldsymbol{\Delta}}\ket{\Phi_{\mathbf{r}}}-\ket{\Phi_{\mathbf{r}+m\mathbf{1}}}\right\Vert\!&\leq m\sqrt{N}\Bigg[4\frac{AD}{B}\left(\frac{\delta}{\sigma}\right)^2+2
\sqrt\frac{\delta }{\sigma}\frac{e^{-\frac{(T-6\delta)^2}{4\sigma^2}}}{B}\Bigg],
\end{align} where $A,B,$ and $D$ are as in \cref{defi:constants}.
\end{lemma}

\begin{proof}[Proof (Adapted from Ref.~\cite{watkins2024time}).]
To do this, it suffices to simply consider the one-qubit case. We have 
$e^{-i\delta \Delta}\ket{\phi_j}= \ket{\phi_j}-i\delta\Delta \ket{\phi_j} + R\ket{\phi_j},$ where \begin{equation}R:= -\int_0^{\delta} \mathrm{d}\tau(\tau e^{-i\tau\Delta }\Delta^2).\end{equation} Note that $\Vert R\ket{\phi_j}\Vert\leq \frac{\delta^2}{2}\Vert \Delta^2\ket{\phi_j}\Vert.$ Up to a phase, $\Delta^2\ket{\phi_j}$ is $\ket{D^2_{\delta}\phi_j},$ so by the triangle inequality and \cref{lemma:doublepartial,lemma:finitedifferenceerror}, we have 
\begin{align}
    \Vert \Delta^2\ket{\phi_j}\Vert &\leq \Vert \ket{D_{\delta}^2\phi_j}-\ket{\partial^2_t\phi_j}\Vert + \Vert \ket{\partial^2_t\phi_j}\Vert\nonumber\\
    &\leq 8\frac{A}{B} \frac{\delta^2}{\sigma^4}+4\frac{A}{B}\frac{1}{\sigma^2}+ \frac{1}{B} \frac{e^{-(T-6\delta)^2/4\sigma^2}}{\sqrt{\sigma\delta^3}}.
\end{align}
Thus 
\begin{equation}\Vert R\ket{\phi_j}\Vert \leq 4\frac{A}{B} \left(\frac{\delta}{\sigma}\right)^4+2\frac{A}{B}\left(\frac{\delta}{\sigma}\right)^2+ \frac{1}{B} \sqrt\frac{\delta}{\sigma}e^{-\frac{(T-6\delta)^2}{4\sigma^2}}.\end{equation}
Next, we note that
\begin{align}
    -i\Delta \ket{g} &= \frac{U_v-U_v^\dagger}{2\delta}\ket{g}\nonumber\\
    &= \sum_{k\in\mathbb{Z}_{N_c}} g(k\delta) \frac{\ket{k+1}-\ket{k-1}}{2\delta}\nonumber\\
    &= - \sum_{k\in\mathbb{Z}_{N_c}}\frac{g((k+1)\delta)-g((k-1)\delta)}{2\delta}\nonumber\\
    &= -\ket{D_{\delta}g}.
\end{align} Thus, by \cref{lemma:finitedifferenceerror}, we have  
\begin{align}
    \Vert e^{-i\delta \Delta} \ket{\phi_j}-(\ket{\phi_j}-\delta\ket{\partial_t\phi_j})\Vert  &\leq \Vert R\ket{\phi_j}\Vert + \sqrt{2}\frac{A}{B}\left(\frac{\delta}{\sigma}\right)^3.
\end{align}
We then have by \cref{lemma:firstordershift} that
\begin{align}\Vert e^{-i\delta \Delta} \ket{\phi_j}-\ket{\phi_{j+1}}\Vert\leq  4\frac{A}{B}\left(\frac{\delta}{\sigma}\right)^2\left[1 + 2^{-3/2}\frac{\delta}{\sigma} +\left(\frac{\delta}{\sigma}\right)^2\right]
+2\sqrt\frac{\delta}{\sigma}\frac{e^{-\frac{(T-6\delta)^2}{4\sigma^2}}}{B}.\end{align}

By repeated application of the triangle inequality, we get the factor of $m,$ and by \cref{lemma:tensor}, we get the extra factor of $\sqrt{N}.$
\end{proof}

\begin{lemma}\label{lemma:controlledapperro} 
Let $\ket{\psi}$ be an arbitrary $N$-qubit state, and define all other terms as above. Then 
    \begin{equation}\Vert e^{-iC(H)\eta }\ket{\psi}\ket{\Phi_{j\mathbf{1}}}-(e^{-iH(j\delta)\eta}\ket\psi)\otimes\ket{\Phi_{j\mathbf{1}}}\Vert\leq \frac{\eta \tau}{2}h_1
    + 2\eta h\sqrt{N\frac{\sigma}{\tau}}\frac{e^{-\frac{\tau^2}{4\sigma^2}}}{B}\end{equation} where $h:= \left[\sum_{{ x}\in E{ \cup V}}\max_t\Vert H_{ x}(t)\Vert\right]$ and $h_1:=\left[\sum_{{ x}\in E{ \cup V}} \max_t\Vert\dot{H}_{ x}(t)\Vert \right].$
\end{lemma}

\begin{proof}[Proof (Adapted from \cite{watkins2024time}).]
Note that 
\begin{equation}C(H) = \sum_{\mathbf{k}\in \mathbb{Z}_{N_c}^{N}}\left(\sum_{e\in E}H_{e}(k_{\rho(e)}\delta){ +\sum_{v\in V}H_v(k_v\delta)}\right)\otimes \ket{\mathbf{k}}\bra{\mathbf{k}}.\end{equation}

Let $P_j^{(1)}:=\sum_{k\in [j-N_q/2,j+N_q/2]}\ket{k}\bra{k}$ and $P_j= \left(P_j^{(1)}\right)^{\otimes N}$, and observe
that $C(H)P_j$ and $C(H)(I-P_j)$ commute. Now, we show that $C(H)P_j$ is approximately equal to $H(j\delta )\otimes P_j.$ To see that, note that
\begin{align}
    C(H)P_j-H(j\delta)\otimes P_j &= \sum_{\mathbf{k}\in [j-N_q/2,j+N_q/2]^{N}} \left(\sum_{e}(H_{e}(k_{\rho(e)}\delta)-H_{e}(j\delta){ +\sum_{v}(H_v(k_v\delta)-H_v(j\delta))}\right)\otimes \ket{\mathbf{k}}\bra{\mathbf{k}},
\end{align}
so
\begin{align}
    \Vert  C(H)P_j-H(j\delta)\otimes P_j\Vert &\leq \max_{\mathbf{k}\in [j-N_q/2,j+N_q/2]^{N}}\left( \left\Vert\sum_{e}[H_{e}(k_{\rho(e)}\delta)-H_{e}(j\delta)]\right\Vert{ +\left\Vert\sum_v \left[H_v(k_v\delta)-H_v(j\delta)\right]\right\Vert}\right)\nonumber\\
    &\leq \max_{\mathbf{k}\in [j-N_q/2,j+N_q/2-1]^{N}} \frac{\tau}{2}\sum_{ x\in E\cup V}\max_{t\in[\min\{j\delta, k_{\rho(e)}\delta\},\max\{j\delta, k_{\rho(e)}\delta\}]}\Vert \dot{H}_{ x}(t)\Vert\nonumber\\
    &\leq  \frac{\tau}{2}h_1.
\end{align}
Thus, 
\begin{equation}\left\Vert e^{-i\eta C(H)}-e^{-i\eta(H(j\delta)\otimes P_j + C(H)(I-P_j))}\right\Vert \leq  \frac{\eta\tau}{2}h_1.\end{equation}

Next, we note that 
\begin{align}
e^{-i\eta H(j\delta)\otimes P_j}\ket{\psi}\otimes \ket{\boldsymbol{\Phi}_{j\mathbf{1}}} &= \left(e^{-iH(j\delta)\eta }\ket\psi\right) \otimes P_j\ket{\boldsymbol{\Phi}_{j\mathbf{1}}}+\ket\psi\otimes (I-P_j)\ket{\boldsymbol{\Phi}_{j\mathbf{1}}}\nonumber\\
&=e^{-iH(j\delta)\eta }\ket\psi\otimes \ket{\boldsymbol{\Phi}_{j\mathbf{1}}}+((I-e^{-iH(j\delta)\eta})\ket\psi)\otimes (I-P_j)\ket{\boldsymbol{\Phi}_{j\mathbf{1}}}.\end{align} We then can apply norms, getting that \begin{align}
\Vert e^{-i\eta H(j\delta)\otimes P_j}\ket\psi\otimes \ket{\boldsymbol{\Phi}_{j\mathbf{1}}}-\left(e^{-i\eta H(j\delta)}\ket\psi\right)\otimes \ket{\boldsymbol{\Phi}_{j\mathbf{1}}}\Vert &\leq \Vert I-e^{-iH(j\delta)\eta}\Vert \Vert (I-P_j)\ket{\boldsymbol{\Phi}_{j\mathbf{1}}}\Vert \nonumber\\
&\leq \eta h
\Vert (I-P_j)\ket{\boldsymbol{\Phi}_{j\mathbf{1}}}\Vert.
\end{align}
Next, we apply $e^{-i\eta C(H)\otimes (I-P_j)}:$
\begin{align}
    e^{-i\eta C(H)(I-P_j)}(e^{-iH(j\delta)\eta}\ket\psi)\otimes \ket{\boldsymbol{\Phi}_{j\mathbf{1}}} &= (e^{-iH(j\delta)\eta}\ket\psi)\otimes P_j\ket{\boldsymbol{\Phi}_{j\mathbf{1}}}\nonumber\\
    &\qquad\qquad+ e^{-i\eta C(H)}\left[(e^{-iH(j\delta)\eta}\ket\psi)\otimes (1-P_j)\ket{\boldsymbol{\Phi}_{j\mathbf{1}}}\right]\nonumber\\
    &= e^{-iH(j\delta)}\ket\psi \otimes\ket{\boldsymbol{\Phi}_{j\mathbf{1}}}\nonumber\\&\qquad+ (e^{-i\eta C(H)}-I)\left[(e^{-iH(j\delta)\eta}\ket\psi)\otimes (1-P_j)\ket{\boldsymbol{\Phi}_{j\mathbf{1}}}\right].\end{align} We can then take the norm to find that
\begin{align}
    \left\Vert \left(e^{-i\eta C(H)(I-P_j)}-I\right)(e^{-iH(j\delta)\eta}\ket\psi)\otimes \ket{\boldsymbol{\Phi}_{j\mathbf{1}}}\right\Vert &\leq \eta h \Vert (I-P_j)\ket{\boldsymbol{\Phi}_{j\mathbf{1}}}\Vert .
\end{align}
Thus, by the triangle inequality, we have that 
\begin{align}
    \Vert e^{-iC(H)\eta}\ket{\psi}\otimes\ket{\boldsymbol{\Phi}_{j\mathbf{1}}} - (e^{-iH(j\delta)\eta} \ket\psi)\otimes \ket{\boldsymbol{\Phi}_{j\mathbf{1}}}\Vert &\leq \frac{\eta \tau}{2}h_1 + 2\eta h \Vert(I-P_j)\ket{\boldsymbol{\Phi}_{j\mathbf{1}}}\Vert\nonumber\\
    &\leq \frac{\eta \tau}{2}h_1 + 2\eta h\sqrt{N\frac{\sigma}{\tau}}\frac{e^{-\frac{\tau^2}{4\sigma^2}}}{B},
\end{align} where, to get to the last line, we use a combination of \cref{lemma:gausstail,lemma:tensor} to bound $\Vert (I-P_j)\ket{\boldsymbol{\Phi_{j\mathbf{1}}}}\Vert. $
\end{proof}

Finally, we provide a full accounting of the error in the localized WWRL construction.

\begin{thm}\label{thm:errorbounds}
    Let $h,h_1, \tau, A,B,$ and $D$ be defined as above in \cref{defi:constants}. Then,  \begin{align}\left\Vert U(T)\ket{\psi}\ket{\boldsymbol{\Phi}_{\mathbf{0}}
    }-e^{-i\overline{H}T}\ket{\psi}\ket{\boldsymbol{\Phi}_{\mathbf{0}}}\right\Vert&\leq 2h_1\frac{T^2}{2N_p} + 2Th\sqrt\frac{N\sigma}{\tau}\frac{e^{-\frac{\tau^2}{4\sigma^2}}}{B} + 4T\sqrt{N}\frac{AD}{B}\frac{\delta}{\sigma^2}\nonumber\\
        &\quad+ 2\sqrt\frac{N_cT}{\sigma}\frac{e^{-\frac{(T-6\delta)^2}{4\sigma^2}}}{B}.\end{align}
\end{thm}

\begin{proof}
    First, we note that as $e^{-i\overline{H}\eta} = e^{-i\eta(C(H)+\boldsymbol{\Delta})},$ we have 
    \begin{equation}\Vert e^{-i\overline{H}\eta}-e^{-i\eta C(H)}e^{-i\eta \boldsymbol{\Delta}}\Vert \leq \frac{\eta^2}{2}\Vert [C(H),\boldsymbol{\Delta}]\Vert \leq h_1\frac{\eta^2}{2}.\end{equation} Thus, by the triangle inequality, we have 
    \begin{equation}\left\Vert e^{-i\overline{H}T}-\left(e^{-i\tau C(H)}e^{-i\tau \boldsymbol{\Delta}}\right)^{N_p}\right\Vert\leq h_1\frac{T^2}{2N_p}. \label{eq:trotterbound}\end{equation}
    Next, we note that \begin{align}\Vert e^{-i\tau C(H)}e^{-i\tau \boldsymbol{\Delta}}\ket{\psi}\ket{\boldsymbol{\Phi}_{j\mathbf{1}}}-(e^{-i\tau H(j\delta)}\ket\psi)\ket{\boldsymbol{\Phi}_{(j+N_q)\mathbf{1}}}\Vert &\leq h_1\frac{\tau^2}{2}+ 2\tau h\sqrt{N\frac{\sigma}{\tau}}\frac{e^{-\frac{\tau^2}{4\sigma^2}}}{B}+ 4\tau\sqrt{N}\frac{A D}{B}\frac{\delta}{\sigma^2} 
    \nonumber\\
    &\qquad +2\sqrt\frac{N_q\tau}{\sigma}\frac{e^{-\frac{(T-6\delta)^2}{4\sigma^2}}}{B}.\end{align}
    Thus, we have         
    \begin{align}
        \left\Vert \left(e^{-i\tau C(H)}e^{-i\tau\boldsymbol{\Delta}}\right)^{N_p}\ket\psi\ket{\boldsymbol{\Phi}_{\mathbf{0}}}-\left(\prod_{k=0}^{N_p-1}e^{-\tau iH(k\tau)}\ket\psi\right)\ket{\boldsymbol{\Phi}_{\mathbf{0}}}\right\Vert &\leq h_1\frac{T^2}{2N_p} + 2Th\sqrt\frac{N\sigma}{\tau}\frac{e^{-\frac{\tau^2}{4\sigma^2}}}{B} + 4T\sqrt{N}\frac{AD}{B}\frac{\delta}{\sigma^2}\nonumber\\
        &\qquad+ 2\sqrt\frac{N_cT}{\sigma}\frac{e^{-\frac{(T-6\delta)^2}{4\sigma^2}}}{B}.
    \end{align}

    Next, combining this and \cref{eq:trotterbound} with the aid of \cref{lemma:discretizationapprox}, we get that
    \begin{align}\left\Vert U(T)\ket{\psi})\ket{\boldsymbol{\Phi}_{\mathbf{0}}}-e^{-i\overline{H}T}\ket{\psi}\ket{\boldsymbol{\Phi}_{\mathbf{0}}}\right\Vert&\leq 2h_1\frac{T^2}{2N_p} + 2Th\sqrt\frac{N\sigma}{\tau}\frac{e^{-\frac{\tau^2}{4\sigma^2}}}{B} + 4T\sqrt{N}\frac{AD}{B}\frac{\delta}{\sigma^2}+ 2\sqrt\frac{N_cT}{\sigma}\frac{e^{-\frac{(T-6\delta)^2}{4\sigma^2}}}{B}.\end{align}
\end{proof}

Now that we have provided rigorous error bounds for the localized WWRL model in terms of $N_p, N_q, \sigma, T$ and properties of the $T$-periodic Hamiltonians $H$, we would like to determine the asymptotic ancilla count and size of $\sigma$ required for given families of Hamiltonians. This follows directly from the above result.

\begin{cor}\label{cor:ancillacount}     
It is sufficient to choose $N_p=\max\{3,\lceil 8h_1T^2\varepsilon^{-1}\rceil\} ,$ $N_q=\max\{4,\lceil\frac{45\sqrt{N}x^2}{\varepsilon}N_p\rceil\},$ and $\sigma= \tau/x,$ with $x = \sqrt{\max\left\{4\log\left(\frac{20hT\sqrt{N}}{\varepsilon}\left[1+\frac{1}{8}\left(\frac{\varepsilon}{16 hT\sqrt{N}}\right)^{4}\right]\right),\frac{10}{N_p^2}\log\left(\frac{1536N^{1/3}}{\varepsilon^{2}}\right),\frac{12}{N_p^2}\log\left(\frac{20\sqrt{N_p}}{\varepsilon}\right)\right\}}$ to have 
\begin{equation}\Vert (U(T,0)\ket\psi)\ket{\boldsymbol\Phi_{\mathbf{0}}}-e^{-i\overline{H}T}\ket\psi\ket{\boldsymbol{\Phi}_{\mathbf{0}}}\Vert\leq \varepsilon.\end{equation}

\end{cor}

\begin{remark}
If $h=N^{\Theta(1)},$ $h_1=N^{\mathcal{O}(1)},$ and $hT=\Omega({1})$ ---conditions that are satisfied by our examples--- the choices reduce to 
\begin{align}\label{eq:asymptoticancilla1}
    N_p &= \mathcal{O}\left(\frac{h_1T^2}{\varepsilon}\right),\\
    N_q &= \mathcal{O}\left(\frac{\sqrt{N}h_1T^2}{\varepsilon^2}\log\frac{hT\sqrt{N}}{\varepsilon}\right),\\
    x &= \mathcal{O}\left(\sqrt{\log\frac{hT\sqrt{N}}{\varepsilon}}\right)\label{eq:asymptoticancilla2}
\end{align}
in the $N\to\infty$ limit. The assumption on $h$ is not unreasonable; if the interaction graph of $H$ is connected and all the terms have norm bounded from below as $N$ increases, then $h=\Omega(N).$ If, additionally, each of the terms of $H$ has polynomially bounded norm, then $h=N^{\mathcal{O}(1)}.$ {Finally, it is also worth noting that a necessary condition for $\Vert U(T,0)-I\Vert =\Omega(1)$ is for $hT$ to be $\Omega(1).$}
\end{remark}

\begin{proof}
    It is sufficient to choose the parameters so that all the four terms $2h_1\frac{T^2}{N_p},$ $2Th\sqrt{\frac{N\sigma}{\tau}}\frac{e^{-\tau^2/4\sigma^2}}{B} ,$ $4T\sqrt{N}\frac{AD}{B}\frac{\delta}{\sigma^2},$ and $2\sqrt\frac{N_cT}{\sigma} \frac{e^{-\frac{-(T-6\delta)^2}{4\sigma^2}}}{B} $ are each individually less than $\varepsilon/4.$ For the first term to be less than $\varepsilon/4,$ we {find that it is sufficient to take} \begin{equation}N_p\geq \max\{3,8h_1T^2/\varepsilon\}.\end{equation}
 Next, we {require} that $A, B^{-1}\leq \frac{3}{2}$ and $D\leq \frac{5}{4}.$ It is sufficient to let $3\delta \leq \sigma \leq T/3$ for all of these to be satisfied. {We show later that this inequality holds, and thus that our scalings are self-consistent.} Now, recalling that $\tau= T/N_p,$ let $x:=\tau/\sigma.$
    
    Next, for the second term (using the loose bound on $B^{-1}$ from the previous paragraph), we see that it is sufficient to take 
    \begin{align} 3h\sqrt\frac{N}{x}Te^{-x^2/4}&\leq \varepsilon/4,\nonumber\\
    \frac{12hT\sqrt{N}}{\varepsilon}&\leq \sqrt{x}e^{x^2/4},\nonumber\\
   \sqrt{ W_0\left(\left[\frac{12hT\sqrt{N}}{\varepsilon}\right]^4\right)} &\leq x, 
    \end{align}
    where $W_j(x)$ denotes the $j^{\mathrm{th}}$ branch of the Lambert $W$ function. For the third term, we have 
    \begin{align}
        \frac{45}{4}\sqrt{N}\frac{N_p}{N_q}x^2 &\leq \frac{\varepsilon}{4},\nonumber\\
        \max\left\{4,\frac{45\sqrt{N}x^2}{\varepsilon}N_p\right\} &\leq N_q.
    \end{align}
    Finally, for the fourth term, we note that, by construction, $N_pN_q\geq 12.$ Then, we have that
    \begin{align}
      2\sqrt\frac{N_cT}{\sigma}\frac{e^{-\frac{(T-6\delta)^2}{4\sigma^2}}}{B}&\leq  2N_p\sqrt{N_qx}e^{-\frac{N_p^2x^2}{8}}.
    \end{align}
Thus, it is sufficient to consider satisfying
\begin{align}
    N_p\sqrt{N_q x}e^{-\frac{(N_px)^2}{4}}&\leq \frac{\varepsilon}{8}.
\end{align}

    Now, as $N_q\geq \max\left\{4,\frac{45\sqrt{N}x^2}{\varepsilon}N_p\right\},$ we need only choose $N_q= \max\left\{4,\left\lceil\frac{45\sqrt{N}x^2}{\varepsilon}N_p\right\rceil\right\}\leq \max\left\{4,\frac{60\sqrt{N}x^2}{\varepsilon}N_p\right\}$ (replacing the $45$ with $60$ as for any $y\geq 3,$ $\lceil y\rceil\leq 4y/3$). Plugging that in, we find 
    \begin{align}
        N_p\sqrt{N_q x}e^{-\frac{N_p^2x^2}{8}} &\leq \max\left\{2N_p\sqrt{x}, 8 N^{1/4}\varepsilon^{-1/2} \sqrt{N_p^3x^3}\right\}e^{-N_p^2x^2/8}.\end{align} 
        Let us first consider the case when the second term is larger. Then, it is sufficient to have 
\begin{equation}x \geq \frac{1}{N_p}\sqrt{-6W_{-1}\left(-\frac{\varepsilon^2}{1536N^{1/3}}\right)}.\end{equation} Now, suppose the first term is larger. Then, we would need
\begin{align}
    x &\geq \frac{1}{N_p}\sqrt{-2W_{-1}\left(-\frac{1}{2}\left[\frac{\varepsilon}{16\sqrt{N_p}}\right]^4\right)}.
\end{align}

Thus, it is sufficient to have \begin{align}
    x \geq \sqrt{\max\left\{W_0\left(\left[\frac{16 hT\sqrt{N}}{\varepsilon}\right]^4\right), -\frac{6}{N_p^2}W_{-1}\left(-\frac{\varepsilon^2}{1536N^{1/3}}\right),-\frac{2}{N_p^2}W_{-1}\left(-\frac{1}{2}\left[\frac{\varepsilon}{16\sqrt{N_p}}\right]^4\right)\right\}}.
\end{align}
Now, we can use the fact that $W_0(u)<\log(2u+1)$ \cite[Thm. 2.3]{Hoorfar2008} and that (when $u,-v\in (-1/e,0)$) 
\begin{align}
W_{-1}(u)&>\log(-u) -\sqrt{-2(\log(-u)+1)}\nonumber\\
-W_{-1}(-v)&< \left(1+\sqrt{2\frac{1-\frac{1}{\log\frac{1}{v}}}{\log\frac{1}{v}}}\right)\log\frac{1}{v},\end{align} \cite[Thm. 1]{Chatzigeorgiou_2013}. For the first argument, we have a bound of \begin{align}\log\left[2\left[\frac{16hT\sqrt{N}}{\varepsilon}\right]^4+1\right]&<  \log\left(\left[\frac{20hT\sqrt{N}}{\varepsilon}\right]^4\left(1+\frac{1}{2}\left[\frac{\varepsilon}{16hT\sqrt{N}}\right]^4\right)\right)\nonumber\\&\leq 4\log\left(\frac{20hT\sqrt{N}}{\varepsilon}\left[1+\frac{1}{8}\left(\frac{\varepsilon}{16hT\sqrt{N}}\right)^4\right]\right),\end{align}
where to get to the last line we used that $(1+x)^{1/4}\leq 1+x/4.$ 
As $v \leq 1/1536$ for the second argument of the maximum, we see that $\sqrt\frac{1-\frac{1}{\log u^{-1}}}{\log u^{-1}}\leq 2/5.$ Thus we can bound the second term by $\frac{10}{N_p^2}\log\left(\frac{1536N^{1/3}}{\varepsilon^2}\right).$
 For the third argument of the maximum, we have $u\leq 2^{17},$ so we find that $\sqrt\frac{1-\frac{1}{\log u^{-1}}}{\log u^{-1}}\leq 1/3.$ Thus, we can bound the third argument by $\frac{3}{N_p^2}\log \left(2\left[\frac{16\sqrt{N_p}}{\varepsilon}\right]^4\right)\leq \frac{12}{N_p^2}\log \left(\frac{20\sqrt{N_p}}{\varepsilon}\right).$ It is thus sufficient to have 
\begin{align}
    x \geq \sqrt{\max\left\{4\log\left(\frac{20hT\sqrt{N}}{\varepsilon}\left[1+\frac{1}{8}\left(\frac{\varepsilon}{16 hT\sqrt{N}}\right)^{4}\right]\right),\frac{10}{N_p^2}\log\left(\frac{1536N^{1/3}}{\varepsilon^{2}}\right),\frac{12}{N_p^2}\log\left(\frac{20\sqrt{N_p}}{\varepsilon}\right)\right\}}.
\end{align}

    Now, we only need to verify that this satisfies the stipulation that $3\delta\leq \sigma \leq T/3,$ or equivalently, that $\frac{3}{N_p}\leq x\leq N_q/3.$ The lower bound is trivially satisfied due to the second term of the maximization (and that $N\geq 1$). The upper bound follows as well, as $x\leq \sqrt{\frac{\varepsilon N_q}{32N_p\sqrt{N}}}\leq \sqrt{\frac{N_q}{16}}\leq \frac{N_q}{3},$ as $N_q\geq 2.$ Thus we are done.
\end{proof}

\section{Mollifier Convolution} \label{sec:mollifconv}

In this section, we present the details of the mollifier-convolution-based smoothing method. We assume we are given some piecewise-continuous $H_{\text{base}}:[0,T]\to\mathcal{B}(\mathcal{H}).$ 

Let $\hat{H}_{\text{base}}(t):=2H_{\text{base}}(2(t-T/4))\mathbf{1}[t\in [T/4,3T/4]].$ Clearly, $\hat{H}_{\text{base}}(t)$ generates the same evolution as $H_{\text{base}}$ at the end of time $T.$ We smooth this to a Hamiltonian $\tilde{H}^{(s)}(t)$ approximating the evolution of $H_{\text{base}}$ after $T$ but also satisfying $\dot{\tilde{H}}^{(s)}(0)=\dot{\tilde{H}}^{(s)}(T)$ and $\tilde{H}^{(s)}(0)=\tilde{H}^{(s)}(T)=0,$ so $\tilde{H}^{(s)}$ is extendible to a $T$-periodic Hamiltonian and we can apply the localized WWRL construction.

To perform this smoothing, we next define the mollifier $\phi_s(t)$ as
\begin{align}
    \phi_s(t):&=\begin{cases}
       \frac{1}{2\nu s}\exp\left[ \frac{1}{\left(\frac{t}{s}\right)^2-1} \right] & \vert t\vert < s,\\
        0 & \text{otherwise},
    \end{cases}
\end{align}
with $\nu \approx 0.222$. It is easy to verify that $\int_{-\infty}^\infty \mathrm{d}t \, \phi_s(t)=1.$ We also note that \begin{align}\phi'_s(t)&=\begin{cases}-\frac{\sqrt{e}}{\nu s^3}\frac{ t \exp\left[ \frac{1}{\frac{t^2}{s^2} - 1} \right]}{(\frac{t^2}{s^2} - 1)^2}&\vert t\vert < s,\\ 0&\text{otherwise}.\end{cases}\end{align}

Next, we let \begin{align}\tilde{H}^{(s)}(t):= (\phi_s* \hat{H}_{\text{base}})(t).\end{align} As this is a convolution, 
$\dot{\tilde{H}}^{(s)}(t) = (\phi'_s*\hat{H}_{\text{base}})(t).$
Observe that $\tilde{H}^{(s)}$ is a smooth function on $[0,T].$ Furthermore, when $s< T/4,$ we have \begin{align}\tilde{H}^{(s)}(0)=\tilde{H}^{(s)}(T) = \dot{\tilde{H}}^{(s)}(0)=\dot{\tilde{H}}^{(s)}(T)=0.\end{align} Thus, $\tilde{H}^{(s)}$ can be extended to a smooth $T$-periodic Hamiltonian in the same way as the main text, so long as $s<T/4$. We now let $\tilde{U}^{(s)}(t)$ be the unitary generated by $\tilde{H}^{(s)}$ with $\tilde{U}^{(s)}(0)=I$ and $U(t)$ be the unitary generated by $H$ with $U(0)=I.$ We now bound how well $U^{(s)}(t)$ approximates $U(t).$
\begin{lemma}\label{lemma:convolutionerror} When $s\leq T/4,$
    \begin{align}
        \left\Vert \tilde{U}^{(s)}(T)-U(T)\right\Vert &\leq \frac{3}{2}h^2sT.
    \end{align}
\end{lemma}

\begin{proof}[Proof (cf.\ App.\ {II} of \cite{Poulin_2011})]
    We have  
    \begin{align}
       \left\Vert \tilde{U}^{(s)}(T)-U(T) \right\Vert&=\left\Vert  \tilde{U}^{(s)}(T)U^\dagger(T)-I\right\Vert,
    \end{align}
    and
    \begin{align}
       \tilde{U}^{(s)}(T)U^\dagger(T)-I &= i\int_0^{T}\mathrm{d}t\tilde{U}^{(s)}(t)(\hat{H}_{\text{base}}(t)-\tilde{H}^{(s)}(t))U^\dagger(t)\nonumber\\
        &=i\int_0^{T}\mathrm{d}t\int_{-\infty}^\infty\mathrm{d}t'\phi_s(t-t')\tilde{U}^{(s)}(t)(\hat{H}_{\text{base}}(t)-\hat{H}_{\text{base}}(t'))U^\dagger(t).
    \end{align}
    Noting that $\phi_s(t-t')\hat{H}_{\text{base}}(t')=0$ and $\hat{H}_{\text{base}}(t)=0$ when $t\notin [T/4,3T/4]\subset [0,T],$ we can then extend the domain of the first integral to give
    \begin{align}
          \tilde{U}^{(s)}(T)U^\dagger(T)-I       &=i\int_{-\infty}^{\infty}\mathrm{d}t\int_{-\infty}^\infty\mathrm{d}t'\phi_s(t-t')\tilde{U}^{(s)}(t)(\hat{H}_{\text{base}}(t)-\hat{H}_{\text{base}}(t'))U^\dagger(t).
    \end{align}
    By the symmetry of $\phi_s(t-t')$ in $t$ and $t',$ we then can write that 
    \begin{align}
          \tilde{U}^{(s)}(T)U^\dagger(T)-I      &=i\int_{-\infty}^{\infty}\mathrm{d}t\int_{-\infty}^\infty\mathrm{d}t'\phi_s(t-t')\left[\tilde{U}^{(s)}(t)\hat{H}_{\text{base}}(t)U^\dagger(t)-\tilde{U}^{(s)}(t')\hat{H}_{\text{base}}(t)U^\dagger(t'))\right].
    \end{align}
    Taking the norm, we find 
    \begin{align}
         \left\Vert\tilde{U}^{(s)}(T)U^\dagger(T)-I \right\Vert &\leq \int_{-\infty}^\infty\mathrm{d}t\int_{-\infty}^\infty \mathrm{d}t' \phi_s(t-t')\left\Vert\left[\tilde{U}^{(s)}(t)\hat{H}_{\text{base}}(t)U^\dagger(t)-\tilde{U}^{(s)}(t')\hat{H}_{\text{base}}(t)U^\dagger(t'))\right]\right\Vert\nonumber \\
         &\leq \int_{-\infty}^\infty\mathrm{d}t\int_{-\infty}^\infty \mathrm{d}t' \phi_s(t-t')\Vert \hat{H}_{\text{base}}(t)\Vert \left(\left\Vert \tilde{U}^{(s)}(t)-\tilde{U}^{(s)}(t')\right\Vert+\left\Vert U(t)-U(t')\right\Vert\right)\nonumber\\
         &\leq2 \max_{t\in[0,T]}\Vert H_{\text{base}}(t)\Vert \int_{T/4}^{3T/4}\mathrm{d}t\int_{-\infty}^\infty\mathrm{d}t'\phi_s(t-t')\left(\left\Vert \tilde{U}^{(s)}(t)-\tilde{U}^{(s)}(t')\right\Vert+\left\Vert U(t)-U(t')\right\Vert\right),
    \end{align}
    where to get to the second line we used the fact that $\Vert AXB-CXD\Vert \leq\Vert X\Vert\left(\Vert A-C\Vert + \Vert B-D\Vert \right).$ Next, noting that $\left\Vert \tilde{U}^{(s)}(t)-\tilde{U}^{(s)}(t')\right\Vert,\left\Vert U(t)-U(t')\right\Vert\leq 2\vert t-t'\vert \max_{t\in [0,T]}\Vert H_{\text{base}}(t)\Vert,$ we get 
\begin{align}
    \left\Vert\tilde{U}^{(s)}(T)U^\dagger(T)-I \right\Vert &\leq 4\left(\max_{t\in[0,T]}\Vert H_{\text{base}}(t)\Vert \right)^2\int_{T/4}^{3T/4}\mathrm{d}t\int_{-\infty}^\infty\mathrm{d}t'\left\vert t-t'\right\vert\phi_s(t-t')\nonumber\\
    &= 2T\left(\max_{t\in[0,T]}\Vert H_{\text{base}}(t)\Vert \right)^2 \int_{-\infty}^\infty \mathrm{d}u \vert u\vert \phi_s(u).
\end{align} We then note that 
\begin{align}
     \int_{-\infty}^\infty \mathrm{d}u \vert u\vert \phi_s(u) &= \frac{1}{2\nu s}\int_{-s}^s \mathrm{d}u\vert u\vert \exp \frac{1}{\left(\frac{u}{s}\right)^2-1}\nonumber\\
     &= \frac{s}{\nu}\int_{-1}^1\mathrm{d}v \vert v\vert \exp\frac{1}{v^2-1}\nonumber\\
     &= \mu s,
\end{align}
where $\mu\approx 0.669.$ Thus, we have 
\begin{align}
    \left\Vert\tilde{U}^{(s)}(T)U^\dagger(T)-I \right\Vert &\leq 2\mu sT \left(\max_{t\in[0,T]}\Vert H_{\text{base}}(t)\Vert \right)^2\nonumber\\
    &\leq \frac{3}{2}h^2sT
\end{align}
as claimed.
\end{proof}

Clearly, \begin{align}\tilde{h}^{(s)}:&=\sum_{{ x\in V\cup E}} \max_t\Vert \tilde{H}^{(s)}_{ x}(t)\Vert\nonumber\\&\leq \sum_{e\in E}\max_t\Vert \hat{H}^{(s)}_{\text{base},{ x}}(t)\Vert = 2h.\end{align} We also prove the following.

\begin{lemma} \label{lemma:h1tilde}
    \begin{align}\tilde{h}_1^{(s)}:=\sum_{{ x\in V\cup E}} \max_t\Vert \dot{\tilde{H}}^{(s)}_{ x}(t)\Vert \leq \frac{3h}{2\nu s}.\end{align}
\end{lemma}

\begin{proof}
    We have $\dot{\tilde{H}}=\phi'_s*H,$ so 
    \begin{align*}
        \Vert \dot{\tilde{H}}(t)\Vert &\leq \frac{2\sqrt{e}}{\nu s^3}h\int_{t-s}^{t+s}\mathrm{d}t' \frac{ \vert t-t'\vert  \exp\frac{1}{\frac{(t-t')^2}{s^2} - 1}}{(\frac{(t-t')^2}{s^2} - 1)^2}\\
        &= \frac{2\sqrt{e}}{\nu s}h \int_{-1}^1\mathrm{d}x \frac{\vert x\vert e^{(x^2-1)^{-1}}}{(x^2-1)^2}\\
        &\leq \frac{3h}{2\nu s}
    \end{align*}
    as claimed.
\end{proof}

Thus, we have the following as a corollary of \cref{thm:errorbounds} (obtained by using the above bounds on $\tilde{h}_1^{(s)}$ and $\tilde{h}^{(s)},$ and adding the error from the convolution).
\begin{cor}\label{cor:mollconvclockerror} Suppose $s\leq T/4.$ Then
    \begin{align}\left\Vert \left(U(T)\ket\psi\right)\ket{\boldsymbol{\Phi}_{\mathbf{0}}}- e^{-i\overline{H}T}\ket\psi \ket{\boldsymbol{\Phi}_{\mathbf{0}}} \right\Vert &\leq \frac{3hT^2}{2N_p\nu s}+ 4Th\sqrt\frac{N\sigma}{\tau}\frac{e^{-\frac{\tau^2}{4\sigma^2}}}{B}+4T\sqrt{N}\frac{AD}{B}\frac{\delta}{\sigma^2}\nonumber\\
    &\qquad + 2\sqrt\frac{N_cT}{\sigma}\frac{e^{-\frac{(T-6\delta)^2}{4\sigma^2}}}{B}+\frac{3}{2}h^2sT.\end{align}
\end{cor}

To make this error below $\varepsilon,$ it suffices to take $s=\min\{T/4, \varepsilon/3h^2T\},$ and then, assuming $h=\Omega(N)$ and $T=\Omega(1/N)$---which all our examples satisfy---take
\begin{align}
    N_c &= \mathcal{O}\left(\frac{\sqrt{N}(hT)^6}{\varepsilon^5}\log\frac{hT\sqrt{N}}{\varepsilon}\right). \label{eq:mollifierancilladim}
\end{align}

\section{The Long-Range Protocol \texorpdfstring{($\alpha\in(d,2d+1)$)}{alpha in (d,2d+1)}}\label{section:LRProt}

In this section, we give a more detailed explanation of the long-range interacting modified GHZ-state preparation protocol. The spatial dimension is $d\in\mathbb{N},$ and we consider some power-law exponent $\alpha\in (d,2d+1].$ 
In particular, we perform the analysis in the regions $\alpha\in (d,2d)$ and $\alpha=2d,$ in addition to the $\alpha\in (2d,2d+1)$ case discussed in the main text. 

Our system has $N=\left(\prod_{j=1}^{s} m_j\right)^d$ ``data qubits'', with $m_j$ specified below, positioned on the integer lattice $\left\{0,1,\dots, N^{1/d}-1\right\}^d\subseteq \mathbb{Z}^d,$ and a base Hamiltonian 
\begin{equation}
    H(t) = \sum_{\mathbf{i}\neq\mathbf{j}}H_{\{\mathbf{i},\mathbf{j}\}}(t) + \sum_{\mathbf{i}}H_{\mathbf{i}}(t),
\end{equation} so that for all $\mathbf{i}\neq\mathbf{j},$  $\Vert H_{\{\mathbf{i},\mathbf{j}\}}\Vert \leq \Vert \mathbf{i}-\mathbf{j}\Vert^{-\alpha}$ and $\Vert H_{\mathbf{i}}\Vert\leq 1.$ 

We now describe a method of constructing $H(t)$ that prepares the state $\ket{\text{GHZ}_N(a,b)}=a\ket{0}^{\otimes N}+b\ket{1}^{\otimes N}$ time-optimally, given the initial state consisting of a tensor product of $a\ket{0}+b\ket{1}$ on site $\mathbf{0}$ and $\ket{0}$ on all other sites. As discussed in the main text, if we can staticize this protocol, we can also construct a time-independent state-transfer protocol with the same asymptotic performance. We partition our sublattice into a collection $\mathcal{C}_j$ of progressively coarser hypercubes, each of size $V_j:= r_j^d,$ where $r_j:= \prod_{k=1}^j m_k,$ with $r_0:=1,$ and 
\begin{align}
    m_j&= \begin{cases} \left\lceil r_{j-1}^{\frac{2d}{\alpha}-1}\right\rceil(1-\delta_{j,1})+3\delta_{j,1}& \alpha \in (d,2d)\\
    \left\lceil \exp\left({\frac{3\sqrt{\log r_{j-1}}}{2\sqrt{d}}}\right) \right\rceil(1-\delta_{j,1})+\left\lceil\exp \frac{8}{d}\right\rceil\delta_{j,1} & \alpha =2d\\
    \left\lceil3^{\frac{1}{\alpha-2d}}\right\rceil+1&\alpha \in (2d,2d+1].
    \end{cases}
\end{align}
Let $C_j\in\mathcal{C}_j$ be the hypercube containing $\mathbf{0}.$ Our protocol is the same for all $\alpha$ regimes, only modifying the value of $m_j.$ Thus, we briefly review the discussion in the main text. We construct our modified GHZ states iteratively, taking $\ket{\mathrm{GHZ}_{V_j}(a,b)}\otimes \bigotimes_{D\in \mathcal{C}_j\backslash\{C_j\}}\ket{\mathrm{GHZ}_{V_j}}_D,$ and producing 
$\ket{\mathrm{GHZ}_{V_{j+1}}(a,b)}\otimes \bigotimes_{D\in \mathcal{C}_{j+1}\backslash\{C_{j+1}\}}\ket{\mathrm{GHZ}_{V_{j+1}}}_D,$ where for brevity we let $\ket{\mathrm{GHZ}_{N}}:=\ket{\mathrm{GHZ}_N(1/\sqrt{2},1/\sqrt{2})}.$

Suppose we are given a protocol to construct the state $\ket{\mathrm{GHZ}_{V_j}(a,b)}_{C_{j}}\otimes \bigotimes_{D\in \mathcal{C}_j\backslash\{C_j\}} \ket{\mathrm{GHZ}_{V_j}}_{D}$ using a sequence $X_j$ of Hamiltonians {(expressed as a tuple, acting from left to right)} applied for times $\mathbf{t}_j.$ Further, let $\cdot$ denote concatenation of {tuples}, and if $A=(a_0,a_1,\dots,a_{\vert Y\vert-1}),$ let $\tilde{A}:=(-a_{\vert Y\vert -1},-a_{\vert Y\vert -2},\dots, -a_1,-a_0),$ applied for times $\tilde{\boldsymbol{t}}:=(t_{\vert Y\vert -1},\dots, t_1, t_0)$.  We then {define} 
\begin{equation}X_{j+1}:=X_j\cdot(H^{(j)}_{2})\cdot \tilde{X}_j\cdot (H_3^{(j)})\cdot X_j,\end{equation}
where we define $H_2^{(j)}$ and $H_3^{(j)}$ below \cite{tran2021optimal}. 

Given $C\in\mathcal{C}_{j+1}\backslash\{C_{j+1}\},$ let $\chi(C)\in \mathcal{C}_{j}$ be an arbitrary subhypercube satisfying $\chi(C)\subset C,$ and let $\chi(C_{j+1}):=C_j.$ Given $B\in\mathcal{C}_{j+1},$ let 
\begin{equation}\Sigma(B):=\{C\in\mathcal{C}_j\backslash\{\chi(B)\}\ :\ C\subset B\}\end{equation} be the set of all hypercubes one step down that are contained in $B,$ other than $\chi(B).$ {For a cartoon illustration of this notation, see \cref{fig:Minhsprotocol}.}  

\begin{figure}
    \centering
    \includegraphics[width=0.5\linewidth]{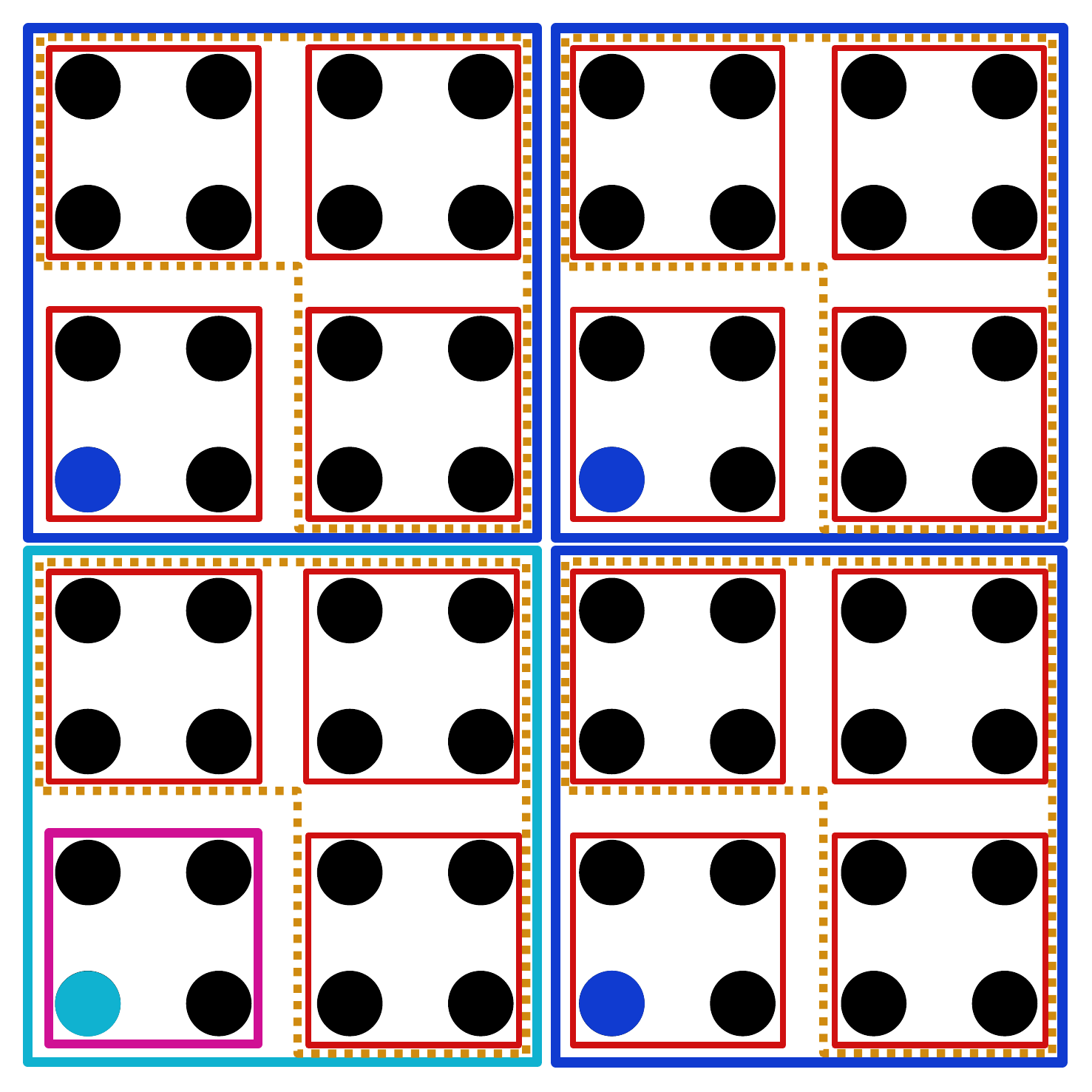}
    \caption{{Cartoon illustration of the setup of the long-range protocol, where $r_j=2$. The smaller squares in red and magenta, form the collection $\mathcal{C}_1,$ and the blue and cyan squares form the collection $\mathcal{C}_2.$ The cyan and magenta squares are $C_2$ and $C_1$ respectively. The dotted orange regions demarcate the squares in $\Sigma(B)$ for the $B$s in $\mathcal{C}_2,$ and thus the squares not in them correspond to $\chi(B).$ {Note that the choice of $\chi(B)$ is arbitrary other than the requirement that $\chi(C_{j+1})=C_j$. For each square $B\in C_2,$ the dot in the square sharing the square's color represents $\varrho(B).$ Note that beyond requiring $\varrho(B)\in B,$ the choice of $\varrho(B)$ is arbitrary as well.}}}
    \label{fig:Minhsprotocol}
\end{figure}
Our Hamiltonian $H_2^{(j)},$ when applied for time $t_2^{(j)}$, will perform, for each $B\in\mathcal{C}_{j+1},$ controlled phase gates on all of $\Sigma(B),$ controlled by $\chi(B),$ in the logical subspace spanned by $\bigcup_{C\in \mathcal{C}_j\cap \mathcal{P}(B)}\{\ket{\overline{0}}_{C},\ket{\overline{1}}_C\}.$  Define \begin{equation}
    H_B:= \sum_{C\in \Sigma(B)}\sum_{(\mu,\nu)\in\chi(B)\times C}\ket{1}\bra{1}_\mu\otimes\ket{1}\bra{1}_\nu,
\end{equation} which performs the controlled phase gates on $B\in \mathcal{C}_{j+1}$ when applied for time $\pi/V_j^2.$ 
We then define 
\begin{equation}H_2^{(j)}:=\frac{1}{\left(r_{j+1}\sqrt{d}\right)^\alpha}\sum_{B\in\mathcal{C}_{j+1}}H_B,\end{equation} which performs $H_B$ for all $B\in\mathcal{C}_{j+1}$ in parallel, and is rescaled to obey the long-range interacting constraint. The time we apply $H_2^{(j)}$ for to get the controlled-phase gates is $t_2^{(j)}:= \pi d^{\alpha/2}r_{j+1}^\alpha/V_j^2.$ We also have $H_3^{(j)},$ which, when applied for time $\pi$ applies Hadamard gates to a single arbitrary site for every hypercube in $\mathcal{C}_j$ other than $C_j,$ which is the hypercube containing $\mathbf{0}.$ Let this assignment of arbitrary sites be denotated by $\varrho:\bigcup_j \mathcal{C}_j\to \{0,1,\dots, N^{1/d}-1\}^{d},$ any function satisfying $\varrho(C)\in C$ for all $C\in \bigcup_j \mathcal{C}_j.$ We then have  \begin{equation}H_3^{(j)}:=\sum_{C\in\mathcal{C}_j\backslash\{C_j\}}\ket{-}\bra{-}_{\varrho(C)}.\end{equation} 

To fully analyze the performance of the localized WWRL construction in this model, we need to calculate $h.$ As $H_2^{(j)}$ and $H_3^{(j)}$ are just sums of commuting projectors with uniform coefficients on the terms, we have that 
$h= \max_j\max\{\Vert H_2^{(j)}\Vert, \Vert H_3^{(j)}\Vert\}.$ Further exploiting that both $H_2^{(j)}$ and $H_3^{(j)}$ are sums of commuting projectors, it is easy to see that they have norms $\frac{N(1-m_{j+1}^{-d})}{V_j^{\frac{\alpha}{d}-1}m_{j+1}^{\alpha }d^{\alpha/2}}$
and $NV_j^{-1}-1$, respectively. Thus, we have 
\begin{equation}
    h = N\max_{j\in \mathbb{Z}_s}\left[V_{j}^{-1}\max\left\{1-\frac{V_j}{N}, \frac{V_j^{2-\frac{\alpha}{d}}(1-m_{j+1}^{-d})}{m_{j+1}^\alpha d^{\alpha/2}}\right\}\right].
\end{equation}
By taking $j=0$ and thus $V_j=1,$ we see that $h=\Omega(N).$
Thus, we need only confirm that $h$ is $\mathcal{O}(N)$ as well. Note that, when $\alpha \geq 2d,$ the second term inside the second maximum is always less than $1$ for both $h$ and $h_1,$ so $h=\Theta(N).$ Now, we turn our focus to the $\alpha\in (d, 2d)$ case. Note that for $\alpha \in (d,2d),$ we can write $m_{j+1}=\left\lceil \frac{3}{2}V_j^{2/\alpha-1/d}\right\rceil\geq V_j^{2/\alpha-1/d},$ so $V_j^{2-\alpha/d}m_{j+1}^{-\alpha}\leq 1.$ Thus, we see that $h=\Theta(N)$ in that case as well.

To finish this analysis, we need only bound the run-time of these algorithms and show that, despite including the time of the Hadamard step, they still scale correctly. From Refs.~\cite{tran2021optimal,Hastings_2006,Chen_2019,Tran_2021_power}, we see that 
\begin{equation}T=\begin{cases}\Omega(\log(r))&\alpha\in(d,2d]\\
\Omega(r^{\alpha-2d})&\alpha\in(2d,2d+1).\end{cases}\end{equation} To find an upper bound, we verify that limiting the speed at which the Hadamard step can be performed does not affect the asymptotic run times discussed in Ref.~\cite{tran2021optimal}. Let $T_j$ be the time to construct the GHZ states in hypercubes of side-length $r_j.$ 
Consider the three cases of $\alpha\in (d,2d),$ $\alpha=2d,$ and $\alpha\in (2d,2d+1).$ For the first case, to perform the induction, we assume that it takes time $T_j\leq K_\alpha \log^{\kappa_\alpha}r_j$ to make a hypercube of side-length $r_j.$ Then, making a hypercube of size-length $r_{j+1}$ takes time 
\begin{align}T_{j+1} &= 3T_j + \pi(1+d^{\alpha/2} m_{j+1}^\alpha r_j^{\alpha-2d}) \nonumber\\
&\leq 3K_\alpha \log^{\kappa_\alpha}r_j + \pi(1+(2(1-\delta_{j,0})+3^\alpha\delta_{j,0})d^{\alpha/2}).\end{align}
Let $j=0.$ We clearly have that $T_0=0.$ Then, we have that $T_{1}\leq \pi(1+3^\alpha d^{\alpha/2}).$ We want this to also be less than or equal to $K_\alpha \log^{\kappa_\alpha} 3.$ To do this, we see that we want $K_\alpha = \pi \frac{1+d^{\alpha/2}3^\alpha}{\log^{\kappa_\alpha}3}.$ We will determine $\kappa_\alpha$ below. Let $\lambda:=2d/\alpha.$ To do so, consider $j> 0.$ We have that 
\begin{align}
T_{j+1}&\leq  K_\alpha \left(3\log^{\kappa_\alpha}r_j +\log^{\kappa_\alpha}3\right)\nonumber\\
&\leq 4K_\alpha \log^{\kappa_\alpha}r_j\nonumber\\
&\leq \frac{4K_\alpha}{\lambda^{\kappa_\alpha}}\log^{\kappa_\alpha}r_{j+1},
\end{align}
which is in turn less than or equal to $K_\alpha\log^{\kappa_\alpha}r_{j+1}$ if we let $\kappa_\alpha := \log_\lambda(4).$

For the case $\alpha=2d,$ we perform similar analysis. Assume $T_j\leq K_\alpha e^{\gamma\sqrt{\log{r_j}}},$  where $\gamma = 3\sqrt{d}.$ Then, let us first consider $j={1},$ where we have
\begin{align}
    T_1 &=  \pi\left(1+d^{\alpha/2}\left\lceil e^{\frac{8}{d}}\right\rceil^\alpha\right).
\end{align}
For this to be at most $K_{\alpha} e^{\gamma\sqrt{\log r_1}},$ we see it is sufficient to have $K_\alpha \leq \pi e^{-6\sqrt{2}}\left(1+(4d)^d e^{16}\right).$

Now let us consider $j\geq {2}.$
\begin{align}
    T_{j+1} &\leq 3T_j + \pi\left(1+d^{\alpha/2}2^\alpha e^{\gamma\sqrt{\log r_j}}\right)\nonumber\\
    &\leq \left(3K_\alpha +\pi d^{\alpha/2}2^\alpha\right)e^{\gamma\sqrt{\log r_j}}+\pi.
\end{align}

As noted in Ref.~\cite{tran2021optimal}, $\gamma \sqrt{\log r_j}\leq \gamma\sqrt{\log(r_{j+1})}-2.$ Thus, to have $T_{j+1}\leq K_\alpha e^{\gamma\sqrt{\log r_{j+1}}},$ it is sufficient to have 
\begin{equation}
    \frac{\pi(1+ (4d)^d)}{e^2-3} \leq K_\alpha.
\end{equation}

Thus, we can set 
\begin{equation}K_\alpha = \max\left\{\pi(1+(4d)^de^{16}) e^{-6\sqrt{2}}, \frac{\pi((4d)^d+1)}{e^2-3}\right\}.\end{equation}
When $\alpha \in (2d,2d+1),$ we assume $T_{j}\leq K_\alpha r_{j}^{\alpha-2d}.$ Note this is satisfied for $T_0$ regardless of $K_\alpha$ as $T_0=0.$ Then,
\begin{align}
    T_{j+1} &\leq 3K_\alpha r_j^{\alpha-2d} + \pi\left(1+d^{\alpha/2}m_{j+1}^\alpha r_j^{\alpha-2d}\right)\nonumber\\
    &\leq \left[\frac{3K_\alpha}{m_{j+1}^{\alpha-2d}}+\pi\left(\frac{1}{r_{j+1}^{\alpha-2d}}+d^{\alpha/2}m_{j+1}^{2d}\right)\right]r_{j+1}^{\alpha-2d}\nonumber\\
    &\leq \left[\frac{3K_\alpha}{m_{j+1}^{\alpha-2d}}+\pi\left(1+d^{\alpha/2}m_{j+1}^{2d}\right)\right]r_{j+1}^{\alpha-2d}.
\end{align}
For this to be less than $K_\alpha r^{\alpha-2d},$ we need only let $m_{j+1}= m:=\left\lceil3^{\frac{1}{\alpha-2d}}\right\rceil+1$ for all $j$ and 
\begin{equation}K_\alpha \geq \pi\frac{1+d^{\alpha/2}m^{2d}}{1-\frac{3}{m^{\alpha-2d}}}.\end{equation}

Thus, we have confirmed that the analysis is asymptotically the same as in Ref.~\cite{tran2021optimal}, differing only by changing constants of proportionality to account for the fact that we do not allow single-site Hamiltonians to be applied with arbitrary strength.

\section{The Strongly Long-Range Protocol \texorpdfstring{($\alpha\in[0,d]$)}{alpha in [0,d]}}\label{section:SLRProt}
In this section, we describe the staticization of the protocol \cite{yin2024ghz} that gives time-optimal (up to logarithmic factors) state transfer for $0\leq \alpha \leq d,$ with run-time scaling as $\mathcal{O}(\log^2 N/N^{1-\alpha/d})$.

We begin by reiterating the original time-dependent protocol. Once again we reduce state transfer to GHZ preparation. Suppose we have some state $\ket{\psi}=a\ket{0}+b\ket{1}$ on site $\mathbf{0}$ in a hypercube of side-length $N^{1/d},$ and all other sites are in the state $\ket{0}.$ We describe a Hamiltonian 

\begin{align*}
    H(t) = \sum_{\mathbf{i}\neq\mathbf{j}}H_{\{\mathbf{i},\mathbf{j}\}}(t)+ \sum_{\mathbf{i}}H_{\mathbf{i}}(t),
\end{align*}
where $\Vert H_{\{\{\mathbf{i},\mathbf{j}\}}(t)\Vert\leq \mathrm{dist}(\mathbf{i},\mathbf{j})^{-\alpha}$ for $\alpha\in [0,d],$ and $\Vert H_{\mathbf{i}}(t)\Vert \leq N.$ Let $\tilde{t}:= N^{1-\alpha/d}t/\log^2N.$ In this case, we can write the protocol as follows {(delaying the intuitive explanation until the next paragraph)}:
\begin{align}
    H(t) &= \begin{cases}
        (\sqrt{d}N^{1/d})^{-\alpha }\sum_{\mathbf{j},\mathbf{k}\neq \mathbf{0}}X_{\mathbf{j}}Y_{\mathbf{k}}+h.c. & \tilde{t}\in [0,\tau_1],\\
        -(\sqrt{d}N^{1/d})^{-\alpha }\sum_{\mathbf{j}\neq \mathbf{0}}Z_{\mathbf{0}}X_{\mathbf{j}} &\tilde{t}\in [\tau_1, (\tau_1+\theta)],\\
        -(\sqrt{d}N^{1/d})^{-\alpha }\sum_{\mathbf{j},\mathbf{k}\neq \mathbf{0}}X_{\mathbf{j}}Y_{\mathbf{k}}+h.c. & \tilde{t}\in[(\tau_1+\theta),(\tau_1+\theta+\tau_2)],\\
        \frac{N^{1-\alpha/d}}{\log N}\sum_{\mathbf{j}\neq\mathbf{0}}X_{\mathbf{j}}&\tilde{t}\in[(\tau_1+\theta+\tau_2),(\tau_1+\theta+\tau_2+\pi/2)],\\
        \frac{1}{(\sqrt{d}N^{1/d})^{\alpha}\log N}\sum_{\mathbf{j},\mathbf{k}\neq\mathbf{0}}Z_{\mathbf{j}}Z_{\mathbf{k}} & \tilde{t}\in [(\tau_1+\theta+\tau_2+\pi/2), (\tau_1+\theta+\tau_2+9\pi/16)],\\
         -\frac{N^{1-\alpha/d}}{\log N}\sum_{\mathbf{j}\neq\mathbf{0}}Z_{\mathbf{j}}&\tilde{t}\in[(\tau_1+\theta+\tau_2+9\pi/16), (\tau_1+\theta+\tau_2+13\pi/16)],\\        
         (\sqrt{d}N^{1/d})^{-\alpha}\sum_{\mathbf{j},\mathbf{k}\neq \mathbf{0}}X_{\mathbf{j}}Y_{\mathbf{j}}+h.c.&\tilde{t}\in [(\tau_1+\theta+\tau_2+13\pi/16), (\tau_1+\theta+\tau_2+13\pi/16+\tau_3)],
    \end{cases}
\end{align}
where $\tau_1,\tau_2,\tau_3,$ and $\theta$ are all $N$-independent, numerically optimized parameters. We note that while in \cite{yin2024ghz} the infidelity of this protocol is numerically shown to be less than $10^{-3}$ for $N\lesssim 2000,$ a rigorous characterization of the infidelity is still currently unknown, so we omit the error dependence from our ancilla dimension count. We also note that this protocol allows the single-site terms to have strength that increases with $N.$ An interesting open direction would be to see if a fast GHZ preparation protocol could have single-site terms with $N$-independent strength.

We can give an intuitive explanation of the protocol as follows. In the first two piecewise time-independent steps, the state is squeezed on all sites other than $\mathbf{0}$ and a rotation controlled on the $\mathbf{0}$ qubit is applied. The squeezing enables the rotation to separate out the parts of the state that will become the $\ket{0^N}$ and $\ket{1^N}$ quicker than if they were just left as spin-coherent states. Once the two parts are sufficiently separated, they can be ``pulled away" to the antipodal points of the Bloch sphere via step 3 of the protocol. Finally, the last 4 steps serve to unsqueeze the state so that is approximately $\ket{\mathrm{GHZ}(a,b)}.$

To understand the scaling of the ancilla size, we calculate
\begin{align*}
    h &= \max\{2(\sqrt{d}N^{1/d})^{-\alpha}(N-1)^2,(\sqrt{d}N^{1/d})^{-\alpha}(N-1),N^{1-\alpha/d}(N-1)/\log N,(\sqrt{d}N^{1/d})^{-\alpha}(N-1)^2/\log N\}\\
    &= 2 d^{-\alpha/2} N^{2-\alpha/d}(1-N^{-1})^2,
\end{align*} 
for sufficiently large $N.$ Recalling that $T=\mathcal{O}(\log^2 N/N^{1-\alpha/d}),$ plugging these into \cref{eq:mollifierancilladim}, and dropping the $\varepsilon$ dependence for the reasons mentioned above, we find 
\begin{align}
    N_c &= \mathcal{O}\left(N^{13/2}\log^{12}N\log \left[N^{3/2}\log^2N\right]\right).
\end{align}

\section{Piecewise Time-Independent Hamiltonians}\label{sec:bumpfunc}

In this section, we give an alternate smoothing procedure for piecewise time-dependent Hamiltonians that can, in some cases, give better scaling of $h$ and $h_1$ than with mollifier-convolutional smoothing.
To do so, we construct a different Hamiltonian that generates exactly the same unitary $U_{\text{base}}(T)$ as the base Hamiltonian after run-time $T$. 

We do this using \emph{bump functions} \begin{equation}\varphi_{T,t_0}(t):= \begin{cases}
    \frac{1}{\nu}\exp\left[-\frac{1}{1-4\left(\frac{t-t_0}{T}-\frac{1}{2}\right)^2}\right]& t\in[t_0,t_0+T]\\
    0 & t\notin[t_0,t_0+T],
\end{cases}\end{equation} where $\nu:=\int_{0}^{1}\mathrm{d}x\exp\left[-(1-4(t-1/2)^2)^{-1}\right]\approx0.222.$ These are smooth, normalized to $\int_{-\infty}^{\infty}\mathrm{d}t \, \varphi_{T,t_0}(t)=T$, and both they and their derivatives are supported on $[t_0,t_0+T].$ Given a piecewise time-independent protocol consisting of Hamiltonians $H_{\text{base}}^{(i)}$ each being applied for time $T_i$ for $i \in \{0,1,\dots, I\},$ we can then construct the differentiable Hamiltonian \begin{equation}\tilde{H}(t) := \sum_{i=0}^{I}{\varphi}_{T_i,\sum_{j=0}^{i-1}T_j}(t)H_{\text{base}}^{(i)}.\end{equation} $\tilde{H}$ is very similar to the given protocol, essentially just modifying the envelope with which we apply the $H^{(i)}$ so it smoothly transitions between different steps. It is easy to see that $\tilde{H}$ exactly generates $U_{\text{base}}(T)$ after time $T.$

To analyze the performance of the localized WWRL construction on $\tilde{H},$ we calculate $h$ and $h_1.$ We know that $\xi:=\max_t\vert {\varphi}_{T,t_0}(t)\vert \approx 1.657.$ Thus, \begin{align}
    h&= \sum_{ x\in V\cup E} \max_{t}\Vert \tilde{H}_{ x}(t)\Vert= \xi \sum_{ x} \max_i \Vert H_{ x}^{(i)}\Vert.
\end{align} 
Next, we note that $ \vert \dot{\varphi}_{T,t_0} \vert = \frac{\zeta}{T}$ for some constant $\zeta \approx 1.597.$ Therefore, we have that \begin{align}h_1&=\sum_{ x\in V\cup E}\max_t\left\Vert\dot{\tilde{H}}_{ x}(t)\right\Vert
=\zeta \sum_{ x}\max_{i} \frac{\Vert H^{(i)}_{ x}\Vert}{T_i}.\end{align}

We thus have everything we need both to apply the localized WWRL model to piecewise time-independent protocols and to analyze the ancilla count scaling of such staticizations using Eqs.~(\ref{eq:asymptoticancilla1}--\ref{eq:asymptoticancilla2}).

\subsection{Better Clock Dimension for the Examples}

For the long-range protocol when $\alpha\in (d,2d+1)$, we apply our controlled-phase steps for time $t_2^{(j)}=\pi d^{\alpha/2}r^{\alpha}_{j+1}/V_j^2$ and our Hadamard steps for time $\pi,$ so for the bump-function construction, 
\begin{align}
    h_1 &= \zeta\max_j\max\left\{(t_2^{(j)})^{-1}\Vert H_2^{(j)}\Vert, \pi^{-1}\Vert H_3^{(j)}\Vert \right\}\nonumber\\
    &=\zeta\pi^{-1}N\!\max_{j\in\mathbb{Z}_s}\left[V_{j}^{-1}\!\max\left\{1-\frac{V_j}{N},\frac{V_j^{2\left(2-\!\frac{\alpha}{d}\right)}(1\!-\!m_{j+1}^{-d}) }{\left(dm_{j+1}^{2}\right)^{\alpha}}\right\}\right].
\end{align}
By the same considerations as for $h,$ we see that $h_1=\Theta(N)$ as well, although the lower bound is strictly speaking unnecessary for our purposes.

Thus, we can determine the dimensionality of the local ancillas. By Eqs.~(\ref{eq:asymptoticancilla1}--\ref{eq:asymptoticancilla2}), we have
\begin{align}
    N_p &= \mathcal{O}\left(\frac{NT^2}{\varepsilon}\right),\\
    N_q &= \mathcal{O}\left(\frac{N^{3/2}T^2}{\varepsilon^2}\log\frac{N^{3/2}T}{\varepsilon}\right).
\end{align}
Multiplying these two together, we get
\begin{equation}N_c=\mathcal{O}\left( \frac{N^{5/2}T^4}{\varepsilon^{3}}\log\frac{ N^{3/2}T}{\varepsilon}\right).\end{equation} We can then take the logarithm to find the number of ancilla qubits needed per data qubit.

For the strongly long-range protocol, each step is applied for time $\Theta(T),$ so $h_1=\mathcal{O}(h/T)=\mathcal{O}(N^{3-2\alpha/d}\log^2N).$ Plugging this in, we find 
\begin{align}
    N_c&=\mathcal{O}\left(N^{5/2}\log^4N\log[N\log^2N]\right).
\end{align}

Similarly, for the disordered protocol, $h_1=\Theta\left(\max_{i=0}^ {N-1}J^2_{\{i,i+1\}}\right)=\Theta(1).$ Thus, plugging in, multiplying, and recalling that $T=\mathcal{O}(N^{1.1z_c}),$ we find that the local ancilla dimension is
\begin{equation}
    N_c=\mathcal{O}\left(\frac{N^{0.5+4.4z_c}}{\varepsilon^3}\log\frac{N^{0.5+1.1z_c}}{\varepsilon}\right).
\end{equation}
\end{document}